\documentclass[12pt,amsymb,fullpage]{amsart}
\usepackage{amssymb,amscd,pstricks}

\newtheorem{theorem}{Theorem}[section]
\newtheorem{defn}[theorem]{Definition}

\newtheorem{lemma}[theorem]{Lemma}

\newtheorem{eple}[theorem]{Example}
\newtheorem{rmk}[theorem]{Remarks}
\newtheorem{dsc}[theorem]{Discussion}
\newtheorem{nota}[theorem]{Notation}

\newsavebox{\indbin}
\savebox{\indbin}{\begin{picture}(0,0)
\newlength{\gnu}
\settowidth{\gnu}{$\smile$} \setlength{\unitlength}{.5\gnu}
\put(-1,-.65){$\smile$} \put(-.25,.1){$|$}
\end{picture}}

\newcommand{\be}{\begin{enumerate}}
\newcommand{\bd}{\begin{defn}}
\newcommand{\bt}{\begin{theorem}}
\newcommand{\bl}{\begin{lemma}}
\newcommand{\ee}{\end{enumerate}}
\newcommand{\ed}{\end{defn}}
\newcommand{\et}{\end{theorem}}
\newcommand{\el}{\end{lemma}}

\begin{document}
\title{A Fourier Inversion Theorem for Normal Functions}
\author{Tristram de Piro}
\address{Flat 3, Redesdale House, 85 The Park, Cheltenham, GL50 2RP }
 \email{t.depiro@curvalinea.net}
\thanks{}
\begin{abstract}
This paper proves an inversion theorem for the Fourier transform defined in \cite{dep3}, applied to the class of normal functions.\\
\end{abstract}
\maketitle
We recall the definition of the Fourier transform for quasi split normal functions, which includes normal functions, introduced in the paper \cite{dep3}, normalised by the factor ${1\over 2\pi}$ in dimension $2$, and by ${1\over (2\pi)^{3\over 2}}$ in dimension $3$, which we denote by $\mathcal{F}$. The aim of this paper is to prove an inversion theorem for such functions. We first have the following;\\

\begin{lemma}
\label{integrablefirst}
Let $f:\mathcal{R}^{2}\rightarrow\mathcal{R}$ be smooth and quasi split normal, then $\mathcal{F}(f)\in L^{1}(\mathcal{R}^{2})$ and is of rapid decay, in the sense that, for $|\overline{k}|>1$, $k_{1}\neq 0,k_{2}\neq 0$\\

$|\mathcal{F}(f)(\overline{k})|\leq {C_{n}\over |\overline{k}|^{n}}$\\

where $C_{n}\in\mathcal{R}$, $n\in\mathcal{N}$.\\

A similar result holds for smooth quasi split normal $f:\mathcal{R}^{3}\rightarrow\mathcal{R}$, with $\mathcal{F}(f)\in L^{1}(\mathcal{R}^{3})$, and for $|\overline{k}|>1$, $k_{1}\neq 0,k_{2}\neq 0, k_{3}\neq 0$\\

$|\mathcal{F}(f)(\overline{k})|\leq {C_{n}\over |\overline{k}|^{n}}$\\

where $C_{n}\in\mathcal{R}$, $n\in\mathcal{N}$.\\\\
\end{lemma}
\begin{proof}
In dimension $2$, by \cite{dep3}, we have that integration by parts is justified, for $k_{1}\neq 0$, $k_{2}\neq 0$, and we obtain that;\\

$\mathcal{F}(\bigtriangledown^{2}(f))(\overline{k}=-k^{2}\mathcal{F}(f)(\overline{k})$\\

$\mathcal{F}((\bigtriangledown^{2})^{n}f)=-k^{2n}\mathcal{F}(f)(\overline{k})$ $(*)$\\

By the definition of quasi split normality, $(\bigtriangledown^{2})^{n}f$ is of moderate decrease $2n+1$ and smooth, so that for $n\geq 1$, $(\bigtriangledown^{2})^{n}f\in L^{1}(\mathcal{R}^{2})$, and we have the trivial bound;\\

$|\mathcal{F}((\bigtriangledown^{2})^{n}f)|\leq {||(\bigtriangledown^{2})^{n}f||_{L^{1}(\mathcal{R}^{2})}\over 2\pi}=C_{2n}$\\

Rearranging $(*)$, we obtain that, for $|\overline{k}|>1$, $k_{1}\neq 0$, $k_{2}\neq 0$;\\

$|\mathcal{F}(f)(\overline{k})|\leq {C_{2n}\over k^{2n}}\leq {C_{2n}\over |k|^{m}}$, for $1\leq m\leq 2n$.\\

The proof for $f:\mathcal{R}^{3}\rightarrow\mathcal{R}$ is similar, noting that $(\bigtriangledown^{2})^{n}f\in L^{1}(\mathcal{R}^{3})$, for $n\geq 2$, and repeating the argument in three variables.\\

We have that, by the definition of quasi split normality, for $f:\mathcal{R}^{2}\rightarrow\mathcal{R}$, $\{{\partial f\over \partial x},{\partial f\over \partial y}\}$ are of moderate decrease $2$, and smooth, so belong to $L^{3\over 2}(\mathcal{R}^{2})$. By the Haussdorff-Young inequality, using the fact that $1\leq {3\over 2}\leq 2$, we have that $\{\mathcal{F}({\partial f\over \partial x}),\mathcal{F}({\partial f\over \partial y})\}\subset L^{3}(\mathcal{R}^{2})$, in particularly $\{\mathcal{F}({\partial f\over \partial x}),\mathcal{F}({\partial f\over \partial y}),|\mathcal{F}({\partial f\over \partial x})|+|\mathcal{F}({\partial f\over \partial y})|\}\subset L^{3}(B(\overline{0},1))$. A simple integration using polar coordinates, shows that ${1\over k}\in L^{3\over 2}(B(\overline{0},1))$. As above, we have that, for $k_{1}\neq 0,k_{2}\neq 0$;\\

$\mathcal{F}(f)(\overline{k})={\mathcal{F}({\partial f\over \partial x})(\overline{k})\over ik_{1}}={\mathcal{F}({\partial f\over \partial y})(\overline{k})\over ik_{2}}$ $(A)$\\

Observe that;\\

${1\over k}={1\over |k_{1}|}{1\over (1+{k_{2}^{2}\over k_{1}^{2}})^{1\over 2}}={1\over |k_{2}|}{1\over (1+{k_{1}^{2}\over k_{2}^{2}})^{1\over 2}}$\\

and;\\

$1\leq (1+{k_{1}^{2}\over k_{2}^{2}})^{1\over 2}\leq \sqrt{2}$, for $|k_{1}|\leq |k_{2}|$\\

$1\leq (1+{k_{2}^{2}\over k_{1}^{2}})^{1\over 2}\leq \sqrt{2}$, for $|k_{2}|\leq |k_{1}|$\\

so that ${1\over |k_{1}|}\leq {\sqrt{2}\over k}$, for $|k_{2}|\leq |k_{1}|$, ${1\over |k_{2}|}\leq {2\over k}$, for $|k_{1}|\leq |k_{2}|$, the cases being exhaustive, $(B)$. Combining $(A),(B)$, we obtain that;\\

$|\mathcal{F}(f)(\overline{k})|\leq {\sqrt{2}|{\mathcal{F}({\partial f\over \partial x})(\overline{k})|\over k}|}$, for $|k_{2}|\leq |k_{1}|$\\

$|\mathcal{F}(f)(\overline{k})|\leq {\sqrt{2}|{\mathcal{F}({\partial f\over \partial y})(\overline{k})|\over k}|}$, for $|k_{1}|\leq |k_{2}|$\\

$|\mathcal{F}(f)(\overline{k})|\leq \sqrt{2}{max(|\mathcal{F}({\partial f\over \partial x})(\overline{k})|,|\mathcal{F}({\partial f\over \partial y})(\overline{k})|)\over k}$\\

$\leq {\sqrt{2}(|\mathcal{F}({\partial f\over \partial x})(\overline{k})|+|\mathcal{F}({\partial f\over \partial y})(\overline{k})|)\over k}$\\

By Holder's inequality, we have that;\\

 ${\sqrt{2}(|\mathcal{F}({\partial f\over \partial x})(\overline{k})|+|\mathcal{F}({\partial f\over \partial y})(\overline{k})|)\over k}\in L^{1}(B(\overline{0},1))$\\

so that $\mathcal{F}(f)(\overline{k})\in L^{1}(B(\overline{0},1))$. By the rapid decrease of $\mathcal{F}(f)$, for $|\overline{k}|>1$, we have that
$\mathcal{F}(f)(\overline{k})\in L^{1}({\mathcal{R}^{2}\setminus B(\overline{0},1)})$, so that $\mathcal{F}(f)(\overline{k})\in L^{1}(\mathcal{R}^{2})$.\\

For $f:\mathcal{R}^{3}\rightarrow\mathcal{R}$, $\{{\partial f\over \partial x},{\partial f\over \partial y},{\partial f\over \partial z}\}$ are of moderate decrease $2$, and smooth, so belong to $L^{2}(\mathcal{R}^{3})$, and by classical theory;\\

 $\{\mathcal{F}({\partial f\over \partial x}),\mathcal{F}({\partial f\over \partial y}),\mathcal{F}({\partial f\over \partial z}),|\mathcal{F}({\partial f\over \partial x})|+|\mathcal{F}({\partial f\over \partial y})|+|\mathcal{F}({\partial f\over \partial z})|\}\subset L^{2}(\mathcal{R}^{3})$\\

as well. In particular;\\

$\{\mathcal{F}({\partial f\over \partial x}),\mathcal{F}({\partial f\over \partial y}),\mathcal{F}({\partial f\over \partial z}),|\mathcal{F}({\partial f\over \partial x})|+|\mathcal{F}({\partial f\over \partial y})|+|\mathcal{F}({\partial f\over \partial z})|\}\subset L^{2}(B(\overline{0},1))$\\

A simple integration using polar coordinates, shows that ${1\over k}\in L^{2}(B(\overline{0},1))$. As above, we have that, for $k_{1}\neq 0,k_{2}\neq 0, k_{3}\neq 0$;\\

$\mathcal{F}(f)(\overline{k})={\mathcal{F}({\partial f\over \partial x})(\overline{k})\over ik_{1}}={\mathcal{F}({\partial f\over \partial y})(\overline{k})\over ik_{2}}={\mathcal{F}({\partial f\over \partial z})(\overline{k})\over ik_{3}}$ $(AA)$\\

Observe that;\\

${1\over k}={1\over |k_{1}|}{1\over (1+{k_{2}^{2}\over k_{1}^{2}}+{k_{3}^{2}\over k_{1}^{2}})^{1\over 2}}={1\over |k_{2}|}{1\over (1+{k_{1}^{2}\over k_{2}^{2}}+{k_{3}^{2}\over k_{2}^{2}})^{1\over 2}}={1\over |k_{3}|}{1\over (1+{k_{1}^{2}\over k_{3}^{2}}+{k_{2}^{2}\over k_{3}^{2}})^{1\over 2}}$\\

and;\\

$1\leq (1+{k_{1}^{2}\over k_{2}^{2}}+{k_{3}^{2}\over k_{2}^{2}})^{1\over 2}\leq \sqrt{3}$, for $max(|k_{1}|,|k_{3}|)\leq |k_{2}|$\\

$1\leq (1+{k_{2}^{2}\over k_{1}^{2}}+{k_{3}^{2}\over k_{1}^{2}})^{1\over 2}\leq \sqrt{3}$, for $max(|k_{2}|,|k_{3}|)\leq |k_{1}|$\\

$1\leq (1+{k_{1}^{2}\over k_{3}^{2}}+{k_{2}^{2}\over k_{3}^{2}})^{1\over 2}\leq \sqrt{3}$, for $max(|k_{1}|,|k_{2}|)\leq |k_{3}|$\\

so that ${1\over |k_{1}|}\leq {\sqrt{3}\over k}$, for $max(|k_{2}|,|k_{3}|)\leq |k_{1}|$, ${1\over |k_{2}|}\leq {\sqrt{3}\over k}$, for $max(|k_{1}|,|k_{3}|)\leq |k_{2}|$, ${1\over |k_{3}|}\leq {\sqrt{3}\over k}$, for $max(|k_{1}|,|k_{2}|)\leq |k_{3}|$ the cases being exhaustive, $(BB)$. Combining $(AA),(BB)$, we obtain that;\\

$|\mathcal{F}(f)(\overline{k})|\leq {\sqrt{3}|{\mathcal{F}({\partial f\over \partial x})(\overline{k})|\over k}|}$, for $max(|k_{2}|,|k_{3}|)\leq |k_{1}|$\\

$|\mathcal{F}(f)(\overline{k})|\leq {\sqrt{3}|{\mathcal{F}({\partial f\over \partial y})(\overline{k})|\over k}|}$, for $max(|k_{1}|,|k_{3}|)\leq |k_{2}|$\\

$|\mathcal{F}(f)(\overline{k})|\leq {\sqrt{3}|{\mathcal{F}({\partial f\over \partial z})(\overline{k})|\over k}|}$, for $max(|k_{1}|,|k_{2}|)\leq |k_{3}|$\\

$|\mathcal{F}(f)(\overline{k})|\leq \sqrt{3}{max(|\mathcal{F}({\partial f\over \partial x})(\overline{k})|,|\mathcal{F}({\partial f\over \partial y})(\overline{k})|,|\mathcal{F}({\partial f\over \partial z})(\overline{k})|)\over k}$\\

$\leq {\sqrt{3}(|\mathcal{F}({\partial f\over \partial x})(\overline{k})|+|\mathcal{F}({\partial f\over \partial y})(\overline{k})|+|\mathcal{F}({\partial f\over \partial z})(\overline{k})|)\over k}$\\

By the Cauchy-Schwartz inequality, we have that;\\

 ${\sqrt{3}(|\mathcal{F}({\partial f\over \partial x})(\overline{k})|+|\mathcal{F}({\partial f\over \partial y})(\overline{k})|+|\mathcal{F}({\partial f\over \partial x})(\overline{k})|)\over k}\in L^{1}(B(\overline{0},1))$\\

so that $\mathcal{F}(f)(\overline{k})\in L^{1}(B(\overline{0},1))$. By the rapid decrease of $\mathcal{F}(f)$, for $|\overline{k}|>1$, we have that
$\mathcal{F}(f)(\overline{k})\in L^{1}({\mathcal{R}^{3}\setminus B(\overline{0},1)})$, so that $\mathcal{F}(f)(\overline{k})\in L^{1}(\mathcal{R}^{3})$.\\

\end{proof}
\begin{defn}
\label{twodimensions}
Let $f\in C^{\infty}(\mathcal{R}^{2})$ be quasi split normal with ${\partial^{i_{1}+i_{2}}f\over \partial x^{i_{1}}\partial y^{i_{2}}}$ bounded for $0\leq i_{1}+i_{2}\leq 27$. Let $C_{m}=\{(x,y)\in\mathcal{R}^{2}:|x|\leq m,|y|\leq m\}$. Let;\\

$Q_{m}={\mathcal{R}^{2}\setminus (x=m\cup x=-m\cup y=m\cup y=-m)}$\\

$C^{13,14,m}(\mathcal{R}^{2})=\{h:{\partial^{i+j}h\over \partial x^{i}\partial y^{j}}, 0\leq i,j\leq 13,\ define\ continuous\ functions,$\\

${\partial^{i+14}h\over \partial x^{i}\partial y^{14}},{\partial^{i+14}h\over \partial x^{14}\partial y^{i}},0\leq i\leq 13,\ define\ bounded\ functions\ on\ Q_{m}\}$\\

Then we define an inflexionary approximation sequence $\{f_{m}:m\in\mathcal{N}\}$ by the requirements;\\

$(i)$. $f_{m}\in C^{13,14,m}(\mathcal{R}^{2})$\\

$(ii)$. $f_{m}|_{C_{m}}=f|_{C_{m}}$\\

$(iii)$ $f_{m}|_{({\mathcal{R}^{2}\setminus C_{m+{1\over m^{2}}}})}=0$\\

Letting $g_{m}=f_{m}|_{[-m,m]\times [-m-{1\over m^{2}},m+{1\over m^{2}}]}$;\\

$(iv)$. For $|x|\leq m$, for $0\leq i\leq 13$;\\

${\partial^{i} g_{m}\over \partial y^{i}}|_{(x,m)}={\partial^{i} f\over \partial y^{i}}|_{(x,m)}$\\

${\partial^{i} g_{m}\over \partial y^{i}}|_{(x,-m)}={\partial^{i} f\over \partial y^{i}}|_{(x,-m)}$\\

${\partial^{i} g_{m}\over \partial y^{i}}|_{(x,m+{1\over m})}=0$\\

${\partial^{i} g_{m}\over \partial y^{i}}|_{(x,-m-{1\over m})}=0$\\

$(v)$. For $|x|\leq m$\\

if ${\partial^{14}f\over \partial y^{14}}|_{(x,m)}>0$, ${\partial^{14}g_{m}\over \partial y^{14}}|_{V_{x,m}}\geq 0$\\

if ${\partial^{14}f\over \partial y^{14}}|_{(x,m)}<0$, ${\partial^{14}g_{m}\over \partial y^{14}}|_{V_{x,m}}\leq 0$\\

if ${\partial^{14}f\over \partial y^{14}}|_{(x,-m)}>0$, ${\partial^{14}g_{m}\over \partial y^{14}}|_{V_{x,-m}}\geq 0$\\

if ${\partial^{14}f\over \partial y^{14}}|_{(x,-m)}<0$, ${\partial^{14}g_{m}\over \partial y^{14}}|_{V_{x,-m}}\leq 0$\\

The same property as $(iv),(v)$ holding, replacing $f$ and $g_{m}$ with ${\partial^{i}f\over \partial x^{i}}$ and ${\partial g_{m}\over \partial x^{i}}$, for $0\leq i\leq 13$.\\

$(vi)$. For $|y|\leq m+{1\over m^{2}}$, $0\leq i\leq 13$\\

${\partial^{i} f_{m}\over \partial x^{i}}|_{(m,y)}={\partial^{i} g_{m}\over \partial x^{i}}|_{(m,y)}$\\

${\partial^{i} f_{m}\over \partial x^{i}}|_{(-m,y)}={\partial^{i} g_{m}\over \partial x^{i}}|_{(-m,y)}$\\

${\partial^{i} f_{m}\over \partial x^{i}}|_{(m+{1\over m},y)}=0$\\

${\partial^{i} f_{m}\over \partial x^{i}}|_{(-m-{1\over m},y)}=0$\\

$(vii)$ For $|y|\leq m+{1\over m^{2}}$\\

if ${\partial^{14}g_{m}\over \partial x^{14}}|_{(m,y)}>0$, ${\partial^{14}f_{m}\over \partial x^{14}}|_{H_{m,y}}\geq 0$\\

if ${\partial^{14}g_{m}\over \partial x^{14}}|_{(m,y)}<0$, ${\partial^{14}f_{m}\over \partial x^{14}}|_{H_{m,y}}\leq 0$\\

if ${\partial^{14}g_{m}\over \partial x^{14}}|_{(-m,y)}>0$, ${\partial^{14}f_{m}\over \partial x^{14}}|_{H_{-m,y}}\geq 0$\\

if ${\partial^{14}g_{m}\over \partial x^{14}}|_{(-m,y)}<0$, ${\partial^{14}f_{m}\over \partial x^{14}}|_{H_{-m,y}}\leq 0$\\

The same property as $(vi),(vii)$ holding, replacing $f_{m}$ and $g_{m}$ with ${\partial^{i}f_{m}\over \partial y^{i}}$ and ${\partial g_{m}\over \partial y^{i}}$, for $0\leq i\leq 14$.\\

where;\\

$V_{x,m}=\{(x,y)\in\mathcal{R}^{2}:y\in (m,m+{1\over m^{2}})\}$\\

$V_{x,-m}=\{(x,y)\in\mathcal{R}^{2}:y\in (-m-{1\over m^{2}},-m)\}$\\

$H_{m,y}=\{(x,y)\in\mathcal{R}^{2}:x\in (m,m+{1\over m^{2}})\}$\\

$H_{-m,y}=\{(x,y)\in\mathcal{R}^{2}:x\in (-m-{1\over m^{2}},-m)\}$\\

\end{defn}
\begin{defn}
\label{threedimensions}
Let $f\in C^{\infty}(\mathcal{R}^{3})$ be quasi split normal with ${\partial^{i_{1}+i_{2}+i_{3}}f\over \partial x^{i_{1}}\partial y^{i_{2}}\partial z^{i_{3}}}$ bounded for $0\leq i_{1}+i_{2}+i_{3}\leq 40$. Let $C_{m}=\{(x,y,z)\in\mathcal{R}^{2}:|x|\leq m,|y|\leq m,|z|\leq m\}$. Let;\\

$Q_{m}={\mathcal{R}^{3}\setminus (x=m\cup x=-m\cup y=m\cup y=-m\cup z=m\cup z=-m)}$\\

$C^{13,13,14,m}(\mathcal{R}^{3})=\{h:{\partial^{i+j+k}h\over \partial x^{i}\partial y^{j}\partial z^{k}}, 0\leq i,j,k\leq 13,\ define\ continuous\ functions,$\\

${\partial^{i+j+14}h\over \partial x^{i}\partial y^{j}\partial z^{14}},{\partial^{i+j+14}h\over \partial x^{i}\partial y^{14}\partial z^{j}},{\partial^{i+j+14}h\over \partial x^{14}\partial y^{i}\partial z^{j}},0\leq i,j\leq 13,\ define\ bounded\ functions\ on\ Q_{m}\}$\\

Then we define an inflexionary approximation sequence $\{f_{m}:m\in\mathcal{N}\}$ by the requirements;\\

$(i)$. $f_{m}\in C^{13,13,14}(\mathcal{R}^{3})$\\

$(ii)$. $f_{m}|_{C_{m}}=f|_{C_{m}}$\\

$(iii)$ $f_{m}|_{({\mathcal{R}^{3}\setminus C_{m+{1\over m^{3}}}})}=0$\\

$(iv)$. For $0\leq |y|\leq m,0\leq |z|\leq m$, for $0\leq i\leq 13$;\\

${\partial^{i} f_{m}\over \partial x^{i}}|_{(m,y,z)}={\partial^{i} f\over \partial x^{i}}|_{(m,y,z)}$\\

${\partial^{i} f_{m}\over \partial x^{i}}|_{(-m,y,z)}={\partial^{i} f\over \partial x^{i}}|_{(-m,y,z)}$\\

${\partial^{i} f_{m}\over \partial x^{i}}|_{(m+{1\over m},y,z)}=0$\\

${\partial^{i} f_{m}\over \partial x^{i}}|_{(-m-{1\over m},y,z)}=0$\\

$(v)$. For $0\leq |y|\leq m,0\leq |z|\leq m$\\

if ${\partial^{14}f\over \partial x^{14}}|_{(m,y,z)}>0$, ${\partial^{14}f_{m}\over \partial x^{14}}|_{H_{m,y,z}}\geq 0$\\

if ${\partial^{14}f\over \partial y^{14}}|_{(m,y,z)}<0$, ${\partial^{14}f_{m}\over \partial x^{14}}|_{H_{m,y,z}}\leq 0$\\

if ${\partial^{14}f\over \partial y^{14}}|_{(-m,y,z)}>0$, ${\partial^{14}f_{m}\over \partial x^{14}}|_{H_{-m,y,z}}\geq 0$\\

if ${\partial^{14}f\over \partial y^{14}}|_{(-m,y,z)}<0$, ${\partial^{14}f_{m}\over \partial x^{14}}|_{H_{-m,y,z}}\leq 0$\\

$(vi)$. For $0\leq |x|\leq m+{1\over m^{3}}$ $0\leq |z|\leq m$, $0\leq i\leq 13$\\

${\partial^{i} f_{m}\over \partial y^{i}}|_{(x,y,z)}={\partial^{i} f_{m}\over \partial y^{i}}|_{(x,m,z)}$, $m\leq y\leq m+{1\over m}$\\

${\partial^{i} f_{m}\over \partial y^{i}}|_{(x,y,z)}={\partial^{i} f_{m}\over \partial y^{i}}|_{(x,-m,z)}$, $-m-{1\over m}\leq y\leq -m$\\

${\partial^{i} f_{m}\over \partial y^{i}}|_{(x,m+{1\over m^{3}},z)}=0$\\

${\partial^{i} f_{m}\over \partial y^{i}}|_{(x,-m-{1\over m^{3}},z)}=0$\\

$(vii)$ For $0\leq |x|\leq m+{1\over m^{3}}$, $0\leq |z|\leq m$\\

if ${\partial^{14}f_{m}\over \partial y^{14}}|_{(x,m,z)}>0$, ${\partial^{14}f_{m}\over \partial y^{14}}|_{V_{x,m,z}}\geq 0$\\

if ${\partial^{14}f_{m}\over \partial y^{14}}|_{(x,m,z)}<0$, ${\partial^{14}f_{m}\over \partial y^{14}}|_{V_{x,m,z}}\leq 0$\\

if ${\partial^{14}f_{m}\over \partial y^{14}}|_{(x,-m,z)}>0$, ${\partial^{14}f_{m}\over \partial y^{14}}|_{V_{x,-m,z}}\geq 0$\\

if ${\partial^{14}f_{m}\over \partial y^{14}}|_{(x,-m,z)}<0$, ${\partial^{14}f_{m}\over \partial y^{14}}|_{V_{x,-m,z}}\leq 0$\\

$(viii)$. For $0\leq |x|\leq m+{1\over m^{3}}$ $0\leq |y|\leq m+{1\over m^{3}}$, $0\leq i\leq 13$\\

${\partial^{i} f_{m}\over \partial z^{i}}|_{(x,y,z)}={\partial^{i} f_{m}\over \partial z^{i}}|_{(x,y,m)}$, $m\leq z\leq m+{1\over m^{3}}$\\

${\partial^{i} f_{m}\over \partial z^{i}}|_{(x,y,z)}={\partial^{i} f_{m}\over \partial z^{i}}|_{(x,y,-m)}$, $-m-{1\over m^{3}}\leq z\leq -m$\\

${\partial^{i} f_{m}\over \partial z^{i}}|_{(x,y,m+{1\over m^{3}})}=0$\\

${\partial^{i} f_{m}\over \partial z^{i}}|_{(x,y,-m-{1\over m^{3}})}=0$\\

$(ix)$ For $0\leq |x|\leq m+{1\over m^{3}}$, $0\leq |y|\leq m+{1\over m^{3}}$\\

if ${\partial^{14}f_{m}\over \partial z^{14}}|_{(x,y,m)}>0$, ${\partial^{14}f_{m}\over \partial z^{14}}|_{D_{x,y,m}}\geq 0$\\

if ${\partial^{14}f\over \partial z^{14}}|_{(x,y,m)}<0$, ${\partial^{14}f_{m}\over \partial z^{14}}|_{D_{x,y,m}}\leq 0$\\

if ${\partial^{14}f\over \partial z^{14}}|_{(x,y,-m)}>0$, ${\partial^{14}f_{m}\over \partial z^{14}}|_{D_{x,y,-m}}\geq 0$\\

if ${\partial^{14}f\over \partial z^{14}}|_{(x,y,-m)}<0$, ${\partial^{14}f_{m}\over \partial z^{14}}|_{D_{x,y,-m}}\leq 0$\\

where;\\

$H_{m,y,z}=\{(x,y,z)\in\mathcal{R}^{3}:x\in (m,m+{1\over m^{3}})\}$\\

$H_{-m,y,z}=\{(x,y,z)\in\mathcal{R}^{3}:x\in (-m-{1\over m^{3}},-m)\}$\\

$V_{x,m,z}=\{(x,y,z)\in\mathcal{R}^{3}:y\in (m,m+{1\over m^{3}})\}$\\

$V_{x,-m,z}=\{(x,y,z)\in\mathcal{R}^{3}:y\in (-m-{1\over m^{3}},-m)\}$\\

$D_{x,y,m}=\{(x,y,z)\in\mathcal{R}^{3}:z\in (m,m+{1\over m^{3}})\}$\\

$D_{x,y,-m}=\{(x,y,z)\in\mathcal{R}^{3}:z\in (-m-{1\over m^{3}},-m)\}$\\
\end{defn}

We now address the issue of the construction of inflexionary approximation sequences in the $2$ and $3$ dimensional cases.\\

\begin{lemma}
\label{squarecube}
The results of Lemma 0.5 in \cite{dep} hold, replacing the intervals $[m,m+{1\over m}]$ with $[m,m+{1\over m^{2}}]$ and $[m,m+{1\over m^{3}}]$.

\end{lemma}
\begin{proof}
In the proof of Lemma 0.5 in \cite{dep}, observe that the coefficients of the polynomial $p$, depend only on the ${1\over m}$ term, so we can obtain the new coefficients for $p$ by substituting $m^{2}$ or $m^{3}$ for $m$. We then calculate in the ${1\over m^{3}}$ case, that;\\

$h'''(x)=(-360a_{0}m^{15}+O(m^{12}))x^{2}+(288a_{0}m^{18}+O(m^{16}))x$\\

$+(-36a_{0}m^{21}+O(m^{19}))$\\

 which has roots when;\\

 $x\simeq {-288a_{0}+/-176a_{0}m^{18}+O(m^{16})\over -720a_{0}m^{15}+O(m^{12})}=O(m^{3})+O(m)>0$\\

 Clearly, we can then assume that for sufficiently large $m$, $h'''(x)$ has no roots in the interval $[-m-{1\over m^{3}}]\cup[m,m+{1\over m^{3}}]$. For the final calculation, with $|h|_{[m+{1\over m^{3}}]}$, we can replace $m$ by $m^{3}$ throughout the proof, to get the same result, that $|h|_{[m+{1\over m^{3}}]}\leq C$, independently of $m>1$. The case with $m^{2}$ replacing $m$ is left to the reader.
\end{proof}

\begin{lemma}
\label{polynomial22D}
If $[a,b]\subset\mathcal{R}$, with $a,b$ finite, and $\{g,g_{1},g_{2}\}\subset C^{\infty}([a,b])$, then, if $m\in\mathcal{R}_{>0}$ is sufficiently large, there exists $h\in C^{\infty}([m,m+{1\over m^{2}}]\times [a,b])$, with the property that;\\

$h(m,y)=g(y)$, ${\partial h\over \partial x}|_{(m,y)}=g_{1}(y)$, ${\partial^{2}h\over \partial x^{2}}|_{(m,y)}=g_{2}(y)$, $y\in [a,b]$, $(i)$\\

$h(m+{1\over m^{2}},y)={\partial h\over \partial x}(m+{1\over m^{2}},y)={\partial^{2}h\over \partial x^{2}}(m+{1\over m^{2}},y)=0$, $y\in [a,b]$, $(ii)$\\

$|h|_{[m,m+{1\over m^{2}}]\times [a,b]}|\leq C$\\

for some $C\in\mathcal{R}_{>0}$, independent of $m$ sufficiently large, and, if ${\partial^{3}h\over \partial x^{3}}(m,y)>0$, ${\partial^{3} h\over \partial x^{3}}(x,y)>0$, for $x\in [m,m+{1\over m^{2}}]$, and if ${\partial^{3}h\over \partial x^{3}}(m,y)<0$, ${\partial^{3} h\over \partial x^{3}}(x,y)<0$, for $x\in [m,m+{1\over m^{2}}]$, $(*)$. In particularly;\\

$\int_{m}^{m+{1\over m^{2}}}|{\partial^{3} h\over \partial x^{3}}|_{(x,y)}|dx=|g_{2}(y)|$\\

Moreover, for $i\in\mathcal{N}$, ${\partial^{i}h\over \partial y^{i}}$ has the property that;\\

${\partial^{i}h\over \partial y^{i}}(m,y)=g^{(i)}(y)$, ${\partial^{i+1} h\over \partial y^{i}\partial x}|_{(m,y)}=g^{(i)}_{1}(y)$, ${\partial^{i+2}h\over \partial y^{i}\partial x^{2}}|_{(m,y)}=g^{(i)}_{2}(y)$\\

$y\in [a,b]$, $(i)'$\\

${\partial^{i}h\over \partial y^{i}}(m+{1\over m^{2}},y)={\partial^{i+1}h\over \partial y^{i}\partial x}(m+{1\over m^{2}},y)={\partial^{i+2}h\over \partial y^{i}\partial x^{2}}(m+{1\over m^{2}},y)=0$\\

$y\in [a,b]$, $(ii)'$\\

$|{\partial^{i}h\over \partial y^{i}}|_{[m,m+{1\over m^{2}}]\times [a,b]}|\leq C_{i}$\\

for some $C_{i}\in\mathcal{R}_{>0}$, independent of $m$ sufficiently large, and, if ${\partial^{i+3}h\over \partial y^{i}\partial x^{3}}(m,y)>0$, ${\partial^{i+3} h\over \partial y^{i}\partial x^{3}}(x,y)>0$, for $x\in [m,m+{1\over m^{2}}]$, and if ${\partial^{i+3}h\over \partial y^{i}\partial x^{3}}(m,y)<0$, ${\partial^{i+3} h\over \partial y^{i}\partial x^{3}}(x,y)<0$, for $x\in [m,m+{1\over m^{2}}]$, $(**)$. In particularly;\\

$\int_{m}^{m+{1\over m^{2}}}|{\partial^{i+3}h\over \partial y^{i}\partial x^{3}}|_{(x,y)}|dx=|g^{(i)}_{2}(y)|$\\

\end{lemma}
\begin{proof}
For the construction of $h$ in the first part, just use the proof of Lemma \ref{squarecube} and Lemma 0.5 in \cite{dep}, replacing the constant coefficients $\{a_{0},a_{1},a_{2}\}\subset \mathcal{R}$ with the data $\{g(y),g_{1}(y),g_{2}(y)\}$. The properties $(i),(ii)$ are then clear. Noting that $[a,b]$ is a finite interval and $\{g,g_{1},g_{2}\}\subset C^{\infty}([a,b])$, by continuity, there exists a constant $D$, with $max(|g(y)|,|g_{1}(y)|,|g_{2}(y)|:y\in [a,b])\leq D$, so, as in the proof of Lemma \ref{squarecube} and Lemma 0.5 in \cite{dep}, we can use the bound $C=16D+7D+D=24D$, for $m>1$. The proof of $(*)$ follows uniformly in $y$, as in the proof of \ref{squarecube} and Lemma 0.5 in \cite{dep}, for sufficiently large $m$, again using the fact that the data $\{g(y),g_{1}(y),g_{2}(y):y\in [a,b]\}$ is bounded. The next claim is just the FTC again. For the second part, when we calculate ${\partial^{i}h\over \partial y^{i}}$, for $i\in\mathcal{N}$, we are just differentiating the coefficients which are linear in the data $\{g(y),g_{1}(y),g_{2}(y)\}$, so we obtain a function which fits the data $\{g^{(i)}(y),g^{(i)}_{1}(y),g^{(i)}_{2}(y)\}$ and $(i)',(ii)'$ follow. Noting that, for $i\in\mathcal{N}$,  $\{g^{(i)},g^{(i)}_{1},g^{(i)}_{2}\}\subset C^{\infty}([a,b])$, again by continuity, there exists constants $D_{i}$, with $max(|g^{(i)}(y)|,|g^{(i)}_{1}(y)|,|g^{(i)}_{2}(y)|:y\in [a,b])\leq D_{i}$, so, again, as in the proof of Lemma \ref{polynomial2}, we can use the bound $C_{i}=16D_{i}+7D_{i}+D_{i}=24D_{i}$, for $m>1$. The proof of $(**)$ follows uniformly in $y$, for each $i\in\mathcal{N}$, as in the proof of Lemma \ref{squarecube} and Lemma 0.5 in \cite{dep}, for sufficiently large $m$, again using the fact that the data $\{g^{(i)}(y),g^{(i)}_{1}(y),g^{(i)}_{2}(y):y\in [a,b]\}$ is bounded. The last claim is again just the FTC.
\end{proof}
\begin{lemma}{Conjecture}\\
\label{polynomial2higher}
\\
Fix $n\in\mathcal{N}$, with $n\geq 3$. If $m\in\mathcal{R}_{>0}$ is sufficiently large, $\{a_{i}:0\leq i\leq n-1\}\subset\mathcal{R}$, there exists $h\in \mathcal{R}[x]$ of degree $2n-1$, with the property that;\\

$h^{(i)}(m)=a_{i}$, $0\leq i\leq n-1$ $(i)$\\

$h^{(i)}(m+{1\over m})=0$, $0\leq i\leq n-1$ $(ii)$\\

$|h|_{[m,m+{1\over m}]}|\leq C$\\

for some $C\in\mathcal{R}_{>0}$, independent of $m$ sufficiently large, and, if $h^{(n)}(m)>0$, $h^{(n)}(x)|_{[m,m+{1\over m}]}>0$, if $h^{(n)}(m)<0$, $h^{(n)}|_{[m,m+{1\over m}]}<0$. In particularly;\\

$\int_{m}^{m+{1\over m}}|h^{(n)}(x)|dx=|a_{n-1}|$, (\footnote{\label{two} If $a_{0}>0$, $a_{1}>0$, there does not exist a smooth function $h$ on the interval $(m,m+{1\over m})$, with $h(m)=a_{0}$, $h'(m)=a_{1}$, $h(m+{1\over m})=0$, $h'(m+{1\over m})=0$, such that $h''>0$ or $h''<0$. To see this, if $h''>0$, using the MVT, we have that $h'(x)>h'(m)>0$, for $x\in (m,m+{1\over m})$, contradicting the fact that $h'(m+{1\over m})=0$. If $h''<0$, and $h'(x)$ has no roots in the interval $(m,m+{1\over m})$, then as $h'(m)>0$, $h'(x)>0$ on $(m,m+{1\over m})$, and $h$ is increasing on $(m,m+{1\over m})$, so that $h(m+{1\over m})>h(m)=a_{0}>0$, contradicting the fact that $h(m+{1\over m})=0$. Otherwise, if $h'(x)$ has a root in the interval $(m,m+{1\over m})$, as $h''<0$, it attains a maximum at $x_{0}\in (m,m+{1\over m})$. Using the MVT again, we must have that for $y\in (x_{0},m+{1\over m})$, $h'(y)<h'(x_{0})=0$, so that $h'(m+{1\over m})<0$, contradicting the fact that $h'(m+{1\over m})=0$.})\\

The same conjecture applies with ${1\over m^{2}}$ and ${1\over m^{3}}$ replacing ${1\over m}$.\\

\end{lemma}
\begin{proof}
We sketch a proof based on the special case $n=3$, which was shown in Lemma 0.5 of \cite{dep}, leaving the details to the reader, (\footnote{\label{inflexionarycurves} One step requires the verification that for a computable polynomial $r_{n}$ of degree $n-1$, $r_{n}(1)\neq 0$, which is highly unlikely on generic grounds and the fact that $r_{3}(1)\neq 1$, although $r_{2}(1)=1$, see footnote \ref{two}. The geometric idea is that allowing for inflexionary type curves, where we can have points $x_{0,i}\in (m,m+{1\over m})$ for which $h^{(i)}(x_{0,i})=0$, where $2\leq i\leq n-1$, the end conditions can be satisfied while still having $h^{(n)}|_{(m,m+{1\over m})}>0$ or $h^{(n)}|_{(m,m+{1\over m})}<0$. However, you still need to do a concrete calculation, which in the case of verifying the conjecture for all $n\in\mathcal{N}$, $n\geq 3$, would involve finding the exact pattern in the coefficients obtained in the proof of Lemma 0.5 of \cite{dep}. We actually only need the result for some $n\geq 14$ in the rest of this paper.}).  We have that $h(x)=(x-(m+{1\over m}))^{n}p(x)$ where $p(x)$ is a polynomial satisfies condition $(ii)$. Computing the derivatives $h^{(i)}(m)$, for $0\leq i\leq n-1$, we obtain $n$ linear equations involving the unknowns $p^{(i)}(m)$, $0\leq i\leq n-1$, of the form;\\

$\sum_{k=0}^{i}{d_{ik}p^{(k)}(m)\over m^{n-i+k}}=a_{i}$, $(0\leq i\leq n-1)$ $(*)$\\

which we can solve for $p^{(i)}(m)$, $0\leq i\leq n-1$, using the fact that the matrix $(d_{ik})_{0\leq i\leq n-1,0\leq k\leq i}$ is lower triangular and $|d_{ii}|=1$, for $0\leq i\leq n-1$. Then we can take;\\

$p(x)=\sum_{i=0}^{n-1}p^{(i)}(m)(x-m)^{i}$\\

so that $h$ has degree $n+(n-1)=2n-1$. It is clear from $(*)$, that we have;\\

$p^{(i)}(m)=\sum_{k=0}^{i}c_{ik}a_{i-k}m^{n+k}$, $(0\leq i\leq n-1)$\\

where $(c_{ik})_{0\leq i\leq n-1,0\leq k\leq i}$ is a real matrix, so that $p(x)$ has the form;\\

$p(x)=\sum_{i=0}^{n-1}v_{i}x^{i}$ $(**)$\\

where;\\

$v_{n-1-i}=\sum_{k=0}^{n-1}r_{ik}m^{n+k}+\sum_{l=0}^{i}s_{il}m^{2n-1+l}$, $(0\leq i\leq n-1)$\\

for real matrices $(r_{ik})_{0\leq i\leq n-1,0\leq k\leq n-1}$ and $(s_{il})_{0\leq i\leq n-1,0\leq l\leq i}$.\\

It is then clear, using the product rule and $(**)$, that;\\

$h^{(n)}(x)=\sum_{k=0}^{n-1}w_{k}x^{k}$\\

where $w_{k}=z_{k}a_{0}m^{3n-2-k}+O(m^{3n-3-k})$, $(0\leq k\leq n-1)$\\

By homogeneity, it is then clear that the real roots of $h^{(n)}(x)$ are of the form $t_{s_{0}}m+O(1)$, where $t_{s_{0}}\in\mathcal{R}$, $1\leq s_{0}\leq n-1$, and $t_{s_{0}}$ satisfies a polynomial $r(x)$ of degree $n-1$, which is effectively computable for given $n$. We can exclude any roots in the interval $[m,m+{1\over m}]$, for sufficiently large $m$, provided $t_{s_{0}}\neq 1$, for $1\leq s_{0}\leq n-1$, which we can check by showing that $r(1)\neq 0$. We have that;\\

$|h|_{(m,m+{1\over m})}|=|(x-(m+{1\over m}))^{n}p(x)|$\\

$\leq {1\over m^{n}}|\sum_{i=0}^{n-1}p^{(i)}(m)(x-m)^{i}|$\\

$\leq {1\over m^{n}}\sum_{i=0}^{n-1}{|p^{(i)}(m)|\over m^{i}}$\\

$\leq \sum_{i=0}^{n-1}\sum_{k=0}^{i}|c_{ik}|a_{i-k}|{m^{n+k}\over m^{n+i}}$\\

$\leq \sum_{i=0}^{n-1}\sum_{k=0}^{i}|c_{ik}|a_{i-k}|=C$, $(m>1)$\\

The last claim is just the FTC.\\

\end{proof}
\begin{lemma}
\label{polynomial22D14}
If $[a,b]\subset\mathcal{R}$, with $a,b$ finite, $n\geq 3$, and $\{g_{j}:0\leq j\leq n-1\}\subset C^{\infty}([a,b])$, then, if $m\in\mathcal{R}_{>0}$ is sufficiently large, there exists $h\in C^{\infty}([m,m+{1\over m^{2}}]\times [a,b])$, with the property that;\\

${\partial^{(j)} h\over \partial x^{j}}|_{(m,y)}=g_{j}(y)$, $y\in [a,b]$, $(i)$\\

${\partial h^{j}\over \partial x^{j}}(m+{1\over m^{2}},y)=0$, $y\in [a,b]$, $(ii)$\\

$|h|_{[m,m+{1\over m^{2}}]\times [a,b]}|\leq C$\\

for some $C\in\mathcal{R}_{>0}$, independent of $m$ sufficiently large, and, if ${\partial^{n}h\over \partial x^{n}}(m,y)>0$, ${\partial^{n} h\over \partial x^{n}}(x,y)>0$, for $x\in [m,m+{1\over m^{2}}]$, and if ${\partial^{n}h\over \partial x^{n}}(m,y)<0$, ${\partial^{n} h\over \partial x^{n}}(x,y)<0$, for $x\in [m,m+{1\over m^{2}}]$, $(*)$. In particularly;\\

$\int_{m}^{m+{1\over m^{2}}}|{\partial^{n} h\over \partial x^{n}}|_{(x,y)}|dx=|g_{n-1}(y)|$\\

Moreover, for $i\in\mathcal{N}$, ${\partial^{i}h\over \partial y^{i}}$ has the property that;\\

${\partial^{i+j}h\over\partial x^{j}\partial y^{i}}(m,y)=g_{j}^{(i)}(y)$, $y\in [a,b]$, $(i)'$\\

${\partial^{i+j}h\over \partial x^{j}\partial y^{i}}(m+{1\over m^{2}},y)=0$, $y\in [a,b]$, $(ii)'$\\

$|{\partial^{i}h\over \partial y^{i}}|_{[m,m+{1\over m^{2}}]\times [a,b]}|\leq C_{i}$\\

for some $C_{i}\in\mathcal{R}_{>0}$, independent of $m$ sufficiently large, and, if ${\partial^{i+n}h\over \partial y^{i}\partial x^{n}}(m,y)>0$, ${\partial^{i+n} h\over \partial y^{i}\partial x^{n}}(x,y)>0$, for $x\in [m,m+{1\over m^{2}}]$, and if ${\partial^{i+n}h\over \partial y^{i}\partial x^{n}}(m,y)<0$, ${\partial^{i+n} h\over \partial y^{i}\partial x^{n}}(x,y)<0$, for $x\in [m,m+{1\over m^{2}}]$, $(**)$. In particularly;\\

$\int_{m}^{m+{1\over m^{2}}}|{\partial^{i+n}h\over \partial y^{i}\partial x^{n}}|_{(x,y)}|dx=|g^{(i)}_{n-1}(y)|$\\

\end{lemma}
\begin{proof}
For the construction of $h$ in the first part, just use the proof of Lemma \ref{polynomial2higher}, replacing the constant coefficients $\{a_{j}:0\leq j\leq n-1\}\subset \mathcal{R}$ with the data $\{g_{j}(y):0\leq j\leq n-1\}$. The properties $(i),(ii)$ are then clear. Noting that $[a,b]$ is a finite interval and $\{g_{j}:0\leq j\leq n-1\}\subset C^{\infty}([a,b])$, by continuity, there exists a constant $D$, with $max(|g_{j}(y)|:0\leq j\leq n-1, y\in [a,b])\leq D$, so, as in the proof of Lemma 0.5 in \cite{dep}, we can use the bound $C=\sum_{0\leq j\leq n-1}L_{j}D$, for $m>1$. The proof of $(*)$ follows uniformly in $y$, as in the proof of Lemma 0.5 in \cite{dep}, for sufficiently large $m$, again using the fact that the data $\{g_{j}(y):0\leq j\leq n-1,y\in [a,b]\}$ is bounded. The next claim is just the FTC again. For the second part, when we calculate ${\partial^{i}h\over \partial y^{i}}$, for $i\in\mathcal{N}$, we are just differentiating the coefficients which are linear in the data $\{g_{j}(y):0\leq j\leq n-1\}$, so we obtain a function which fits the data $\{g^{(i)}_{j}(y):0\leq j\leq n-1\}$ and $(i)',(ii)'$ follow. Noting that, for $i\in\mathcal{N}$,  $\{g^{(i)}_{j}:0\leq j\leq n-1\}\subset C^{\infty}([a,b])$, again by continuity, there exist constants $D_{i}$, with $max(|g^{(i)}_{j}(y)|:0\leq j\leq n-1,y\in [a,b])\leq D_{i}$, so, again, as in the proof of Lemma 0.5 in \cite{dep}, we can use the bound $C_{i}=\sum_{0\leq j\leq n-1}L_{j}D_{i}$, for $m>1$. The proof of $(**)$ follows uniformly in $y$, for each $i\in\mathcal{N}$, as in the proof of Lemma 0.5 in \cite{dep}, for sufficiently large $m$, again using the fact that the data $\{g^{(i)}_{j}(y):0\leq j\leq n-1,y\in [a,b]\}$ is bounded. The last claim is again just the FTC.
\end{proof}

\begin{lemma}
\label{polynomial23D14}
If $[a,b]\subset\mathcal{R}$, $[c,d]\subset\mathcal{R}$,with $a,b,c,d$ finite, $n\geq 3$, and $\{g_{j}:0\leq j\leq n-1\}\subset C^{\infty}([a,b]\times [c,d])$, then, if $m\in\mathcal{R}_{>0}$ is sufficiently large, there exists $h\in C^{\infty}([m,m+{1\over m^{3}}]\times [a,b]\times [c,d])$, with the property that;\\

${\partial^{(j)} h\over \partial x^{j}}|_{(m,y,z)}=g_{j}(y,z)$, $(y,z)\in [a,b]\times [c,d]$, $(i)$\\

${\partial h^{j}\over \partial x^{j}}(m+{1\over m^{3}},y,z)=0$, $(y,z)\in [a,b]\times [c,d]$, $(ii)$\\

$|h|_{[m,m+{1\over m^{3}}]\times [a,b]\times [c,d]}|\leq C$\\

for some $C\in\mathcal{R}_{>0}$, independent of $m$ sufficiently large, and, if ${\partial^{n}h\over \partial x^{n}}(m,y,z)>0$, ${\partial^{n} h\over \partial x^{n}}(x,y,z)>0$, for $x\in [m,m+{1\over m^{3}}]$, and if ${\partial^{n}h\over \partial x^{n}}(m,y,z)<0$, ${\partial^{n} h\over \partial x^{n}}(x,y,z)<0$, for $x\in [m,m+{1\over m^{3}}]$, $(*)$. In particularly;\\

$\int_{m}^{m+{1\over m^{3}}}|{\partial^{n} h\over \partial x^{n}}|_{(x,y,z)}|dx=|g_{n-1}(y,z)|$\\

Moreover, for $(i,k)\subset \mathcal{N}^{2}$, $0\leq j\leq n-1$, ${\partial^{i+k}h\over \partial y^{i}\partial z^{k}}$,  has the property that;\\

${\partial^{i+j+k}h\over\partial x^{j}\partial y^{i}\partial z^{k}}(m,y,z)={\partial^{i+k} g_{j}\over \partial y^{i}\partial z^{k}}(y,z)$, $(y,z)\in [a,b]\times [c,d]$, $(i)'$\\

${\partial^{i+j+k}h\over \partial x^{j}\partial y^{i}\partial z^{k}}(m+{1\over m^{3}},y,z)=0$, $(y,z)\in [a,b]\times [c,d]$, $(ii)'$\\

$|{\partial^{i+k}h\over \partial y^{i}\partial z^{k}}|_{[m,m+{1\over m^{3}}]\times [a,b]\times [c,d]}|\leq C_{i,k}$\\

for some $C_{i,k}\in\mathcal{R}_{>0}$, independent of $m$ sufficiently large, and, if ${\partial^{i+k+n}h\over \partial y^{i}\partial z^{k}\partial x^{n}}(m,y,z)>0$, ${\partial^{i+k+n} h\over \partial y^{i}\partial z^{k}\partial x^{n}}(x,y,z)>0$, for $x\in [m,m+{1\over m^{3}}]$, and if ${\partial^{i+k+n}h\over \partial y^{i}\partial z^{k}\partial x^{n}}(m,y)<0$, ${\partial^{i+k+n} h\over \partial y^{i}\partial z^{k}\partial x^{n}}(x,y,z)<0$, for $x\in [m,m+{1\over m^{3}}]$, $(**)$. In particularly;\\

$\int_{m}^{m+{1\over m^{3}}}|{\partial^{i+k+n}h\over \partial y^{i}\partial z^{k}\partial x^{n}}|_{(x,y,z)}|dx=|{\partial^{i+k} g_{n-1}\over \partial y^{i}\partial z^{k}}(y,z)|$\\

\end{lemma}
\begin{proof}
For the construction of $h$ in the first part, just use the proof of Lemma \ref{polynomial2higher}, replacing the constant coefficients $\{a_{j}:0\leq j\leq n-1\}\subset \mathcal{R}$ with the data $\{g_{j}(y,z):0\leq j\leq n-1\}$. The properties $(i),(ii)$ are then clear. Noting that $[a,b]\times [c,d]$ is compact and $\{g_{j}:0\leq j\leq n-1\}\subset C^{\infty}([a,b]\times [c,d])$, by continuity, there exists a constant $D$, with $max(|g_{j}(y,z)|:0\leq j\leq n-1, (y,z)\in [a,b]\times [c,d])\leq D$, so, as in the proof of Lemma \ref{polynomial2higher}, we can use the bound $C=\sum_{0\leq j\leq n-1}L_{j}D$, for $m>1$. The proof of $(*)$ follows uniformly in $y$, as in the proof of \ref{polynomial2higher}, for sufficiently large $m$, again using the fact that the data $\{g_{j}(y,z):0\leq j\leq n-1,(y,z)\in [a,b]\}$ is bounded. The next claim is just the FTC again. For the second part, when we calculate ${\partial^{i+k}h\over \partial y^{i}\partial z^{k}}$, for $(i,j\in\mathcal{N}^{2}$, we are just differentiating the coefficients which are linear in the data $\{g_{j}(y,z):0\leq j\leq n-1\}$, so we obtain a function which fits the data $\{{\partial^{i+k}g_{j}\over \partial y^{i}\partial z^{k}}(y,z):0\leq j\leq n-1\}$ and $(i)',(ii)'$ follow. Noting that, for $(i,k)\in\mathcal{N}^{2}$,  $\{{\partial^{i+k}g_{j}\over \partial y^{i}\partial z^{k}}:0\leq j\leq n-1\}\subset C^{\infty}([a,b]\times [c,d])$, again by continuity, there exist constants $D_{i,k}$, with $max(|{\partial^{i+k}g_{j}\over \partial y^{i}\partial z^{k}}(y,z)|:0\leq j\leq n-1,y\in [a,b]\times [c,d])\leq D_{i,k}$, so, again, as in the proof of Lemma \ref{polynomial2higher}, we can use the bound $C_{i,k}=\sum_{0\leq j\leq n-1}L_{j}D_{i,k}$, for $m>1$. The proof of $(**)$ follows uniformly in $(y,z)$, for each $(i,k)\in\mathcal{N}^{2}$, as in the proof of Lemma \ref{polynomial2higher}, for sufficiently large $m$, again using the fact that the data $\{{\partial^{i+k}g_{j}\over \partial Y^{i}\partial z^{k}}(y):0\leq j\leq n-1,(y,z)\in [a,b]\times [c,d]\}$ is bounded. The last claim is again just the FTC.
\end{proof}
\begin{lemma}
\label{existence}
For $f\in C^{\infty}(\mathcal{R}^{2})$ with ${\partial^{i_{1}+i_{2}}f\over \partial x^{i_{1}}\partial y^{i_{2}}}$ bounded by some constant $F\in\mathcal{R}_{>0}$, for $0\leq i_{1}+i_{2}\leq 27$. Then for sufficiently large $m$, there exists an inflexionary approximation sequence $\{f_{m}:m\in\mathcal{N}\}$, with the property that;\\

$max(\int_{\mathcal{R}^{2}}|{\partial f_{m}\over \partial x^{14}}|dxdy,\int_{\mathcal{R}^{2}}|{\partial f_{m}\over \partial y^{14}}|dxdy)\leq Gm^{2}$\\

for some $G\in\mathcal{R}_{>0}$, for sufficiently large $m$.\\

\end{lemma}
\begin{proof}
Define $f_{m}=f$ on $C_{m}$, so that $(ii)$ of Definition \ref{twodimensions} is satisfied. Using two applications of Lemma \ref{polynomial22D14} with $n=14$, changing to a vertical rather than horizontal orientation, and the fact that, for $0\leq i\leq 13$, $|x|\leq m$, ${\partial^{i} f\over \partial y^{i}}|_{(x,m)}$ and ${\partial^{i} f\over \partial y^{i}}|_{(x,-m)}$ define smooth functions on $[-m,m]$, we can extend $f_{m}$ to $R=\{(x,y):|x|\leq m, m\leq |y|\leq m+{1\over m^{2}}\}$, such that $f_{m}|R_{1}$ satisfies conditions $(iv),(v)$ of Definition \ref{twodimensions}, where $R_{1}=\{(x,y):|x|\leq m, 0\leq |y|\leq m+{1\over m^{2}}\}$. Again, using two applications of Lemma \ref{polynomial22D14} with $n=14$, and the original horizontal orientation, and the fact that, for $0\leq i\leq 13$, $0\leq |y|\leq m+{1\over m^{2}}$, ${\partial^{i} f_{m}\over \partial x^{i}}|_{(m,y)}$ and ${\partial^{i} f\over \partial x^{i}}|_{(-m,y)}$ define smooth functions on $[-m-{1\over m^{2}},m+{1\over m^{2}}]$, we can extend $f_{m}$ to $S=\{(x,y):m\leq |x|\leq m+{1\over m^{2}}, 0\leq |y|\leq m+{1\over m^{2}}\}$, such that $f_{m}|C_{m+{1\over m^{2}}}$ satisfies conditions $(vi),(vii)$ of Definition \ref{twodimensions}. Conditions $(i),(iii)$ are then clear. We then have, using $(iii)$, that;\\

$\int_{\mathcal{R}^{2}}|{\partial f_{m}\over \partial x^{14}}|dxdy=\int_{C_{m+{1\over m^{2}}}}|{\partial f_{m}\over \partial x^{14}}|dxdy$\\

$=\int_{|x|\leq m,|y|\leq m}|{\partial f_{m}\over \partial x^{14}}|dxdy+\int_{|x|\leq m,m\leq |y|\leq m+{1\over m^{2}}}|{\partial f_{m}\over \partial x^{14}}|dxdy+\int_{m\leq |x|\leq m+{1\over m^{2}},|y|\leq m}|{\partial f_{m}\over \partial x^{14}}|dxdy$\\

$+\int_{m\leq |x|\leq m+{1\over m^{2}},m\leq |y|\leq m+{1\over m^{2}}}|{\partial f_{m}\over \partial x^{14}}|dxdy$\\

$\int_{\mathcal{R}^{2}}|{\partial f_{m}\over \partial y^{14}}|dxdy=\int_{C_{m+{1\over m^{2}}}}|{\partial f_{m}\over \partial y^{14}}|dxdy$\\

$=\int_{|x|\leq m,|y|\leq m}|{\partial f_{m}\over \partial y^{14}}|dxdy+\int_{|x|\leq m,m\leq |y|\leq m+{1\over m^{2}}}|{\partial f_{m}\over \partial y^{14}}|dxdy+\int_{m\leq |x|\leq m+{1\over m^{2}},|y|\leq m}|{\partial f_{m}\over \partial y^{14}}|dxdy$\\

$+\int_{m\leq |x|\leq m+{1\over m^{2}},m\leq |y|\leq m+{1\over m^{2}}}|{\partial f_{m}\over \partial y^{14}}|dxdy$ $(*)$\\

We then have the following cases, using the second clause in Lemma \ref{polynomial22D14} repeatedly with the appropriate orientations;\\

Case 1;\\

$\int_{|x|\leq m,|y|\leq m}|{\partial^{14}f_{m}\over \partial x^{14}}|dxdy$\\

$=\int_{|x|\leq m,|y|\leq m}|{\partial^{14}f\over \partial x^{14}}|dxdy\leq Fm^{2}$\\

$\int_{|x|\leq m,|y|\leq m}|{\partial^{14}f_{m}\over \partial y^{14}}|dxdy$\\

$=\int_{|x|\leq m,|y|\leq m}|{\partial^{14}f\over \partial y^{14}}|dxdy\leq Fm^{2}$\\

Case 2;\\

$\int_{|x|\leq m,m\leq |y|\leq m+{1\over m^{2}}}|{\partial^{14}f_{m}\over \partial x^{14}}|dxdy$\\

$=\int_{|x|\leq m}(\int_{|y|\leq m+{1\over m^{2}}}|{\partial^{14}f_{m}\over \partial x^{14}}|dy)dx$\\

$\leq {2\over m^{2}}\int_{|x|\leq m}C_{14}dx$\\

$\leq 2m{2\over m^{2}}C_{14}$\\

$=4{C_{14}\over m}$\\

Case 3;\\

$\int_{m\leq |x|\leq m+{1\over m^{2}},|y|\leq m}|{\partial^{14}f_{m}\over \partial x^{14}}|dxdy$\\

$=\int_{|y|\leq m}(\int_{m\leq |x|\leq m+{1\over m^{2}}}|{\partial^{14}f_{m}\over \partial x^{14}}|dx)dy$\\

$=\int_{|y|\leq m}(|{\partial^{13}f\over \partial x^{13}}|_{(m,y)}+|{\partial^{13}f\over \partial x^{13}}|_{(-m,y)})dy$\\

$\leq 4mF$\\

Case 4.\\

$\int_{m\leq |x|\leq m+{1\over m^{2}},m\leq |y|\leq m+{1\over m^{2}}}|{\partial^{14}f_{m}\over \partial x^{14}}|dxdy$\\

$=\int_{m\leq |y|\leq m+{1\over m^{2}}}(\int_{m\leq |x|\leq m+{1\over m^{2}}}|{\partial^{14}f_{m}\over \partial x^{14}}|dx)dy$\\

$=\int_{m\leq |y|\leq m+{1\over m^{2}}}(|{\partial^{13}f_{m}\over \partial x^{13}}|_{(m,y)}+|{\partial^{13}f_{m}\over \partial x^{13}}|_{(-m,y)}dy$\\

$\leq \int_{m\leq y\leq m+{1\over m^{2}}}C_{13,1}dy+\int_{-m-{1\over m^{2}}\leq -m}C_{13,2}dy$\\

$\leq {max(C_{13,1},C_{13,2})\over m^{2}}$ (the constants $\{C_{13,1},C_{13,2}\}$ coming from the two applications of Lemma \ref{polynomial22D14} at the two boundaries)\\

Case 5;\\

$\int_{|x|\leq m,m\leq |y|\leq m+{1\over m^{2}}}|{\partial^{14}f_{m}\over \partial y^{14}}|dxdy$\\

$= \int_{|x|\leq m}(\int_{m\leq |y|\leq m+{1\over m^{2}}}|{\partial^{14}f_{m}\over \partial y^{14}}|dy)dx$\\

$= \int_{|x|\leq m}(|{\partial f\over \partial y^{13}}|_{(x,m)}+|{\partial f(x,y)\over \partial y^{13}}|_{(x,-m)}dx)$\\

$\leq 4mF$\\

Case 6;\\

$\int_{|y|\leq m,m\leq |x|\leq m+{1\over m^{2}}}|{\partial^{14}f_{m}\over \partial y^{14}}|dxdy$\\

$=\int_{|y|\leq m}(\int_{m\leq |x|\leq m+{1\over m^{2}}}|{\partial^{14}f_{m}\over \partial y^{14}}|dx)dy$\\

$\leq {1\over m^{2}}\int_{|y|\leq m}(|\sum_{i=0}^{13}D_{i}|{\partial^{i}\partial^{14}f\over \partial y^{14}\partial x^{i}}|(m,y)+|\sum_{i=0}^{13}D_{i}|{\partial^{i}\partial^{14}f\over \partial y^{14}\partial x^{i}}|(-m,y))dy$\\

$\leq {2\over m^{2}}(2m)F(\sum_{i=0}^{13}D_{i})$\\

$=4F{(\sum_{i=0}^{13}D_{i})\over m}$\\

Case 7.\\

$\int_{m\leq |x|\leq m+{1\over m^{2}},m\leq |y|\leq m+{1\over m^{2}}}|{\partial^{14}f_{m}\over \partial y^{14}}|dxdy$\\

$=\int_{m\leq |y|\leq m+{1\over m^{2}}}(\int_{m\leq |x|\leq m+{1\over m^{2}}}|{\partial^{14}f_{m}\over \partial y^{14}}|dx)dy$\\

$\leq {1\over m^{2}}\int_{m\leq |y|\leq m+{1\over m^{2}}}(\sum_{i=0}^{13}L_{i,14}|{\partial^{i+14}f_{m}\over \partial x^{i}\partial y^{14}}|_{(m,y)}+L_{i,14}|{\partial^{i+14}f_{m}\over \partial x^{i}\partial y^{14}}|_{(-m,y)})dy$\\

$={1\over m^{2}}\sum_{i=0}^{13}L_{i,14}(|{\partial^{i+13}f\over \partial x^{i}\partial y^{13}}|_{(m,m)}|+|{\partial^{i+13}f\over \partial x^{i}\partial y^{13}}|_{(m,-m)}|+|{\partial^{i+13}f\over \partial x^{i}\partial y^{13}}|_{(-m,m)}|+|{\partial^{i+13}f\over \partial x^{i}\partial y^{13}}|_{(-m,-m)}|)$\\

$\leq {4F(\sum_{i=0}^{13}L_{i,14})\over m^{2}}$ (the constants $L_{i,14}$,$0\leq i\leq 13$ coming from the proof of Lemma \ref{polynomial22D14})\\

Combining the seven cases and $(*)$, we obtain, for sufficiently large $m$, that;\\

$\int_{\mathcal{R}^{2}}|{\partial f_{m}\over \partial x^{14}}|dxdy\leq Fm^{2}+4{C_{14}\over m}+4mF+{max(C_{13,1},C_{13,2})\over m^{2}}\leq Gm^{2}$\\

$\int_{\mathcal{R}^{2}}|{\partial f_{m}\over \partial y^{14}}|dxdy\leq Fm^{2}+4mF+4F{(\sum_{i=0}^{13}D_{i})\over m}+{4F(\sum_{i=0}^{13}L_{i,14})\over m^{2}}\leq Gm^{2}$\\

\end{proof}

\begin{lemma}
\label{existence2}
For $f\in C^{40}(\mathcal{R}^{3})$ with ${\partial^{i_{1}+i_{2}+i_{3}}f\over \partial x^{i_{1}}\partial y^{i_{2}}\partial z^{i_{3}}}$ bounded by some constant $F\in\mathcal{R}_{>0}$, for $0\leq i_{1}+i_{2}+i_{3}\leq 40$. Then for sufficiently large $m$, there exists an inflexionary approximation sequence $\{f_{m}:m\in\mathcal{N}\}$, with the property that;\\

$max(\int_{\mathcal{R}^{3}}|{\partial f_{m}\over \partial x^{14}}|dxdydz,\int_{\mathcal{R}^{3}}|{\partial f_{m}\over \partial y^{14}}|dxdydz,\int_{\mathcal{R}^{3}}|{\partial f_{m}\over \partial z^{14}}|dxdydz)\leq Gm^{3}$\\

for some $G\in\mathcal{R}_{>0}$, for sufficiently large $m$.\\

\end{lemma}
\begin{proof}
Define $f_{m}=f$ on $W_{m}$, so that $(ii)$ of Definition \ref{threedimensions} is satisfied. Using two applications of Lemma \ref{polynomial23D14} with $n=14$, with a horizontal orientation, and the fact that, for $0\leq i\leq 13$, $0\leq |y|\leq m$, $0\leq |z|\leq m$ ${\partial^{i} f\over \partial x^{i}}|_{(m,y,z)}$ and ${\partial^{i} f\over \partial x^{i}}|_{(-m,y,z)}$ define smooth functions on $[-m,m]^{2}$, we can extend $f_{m}$ to $A_{1}=\{(x,y,z):m\leq |x|\leq m+{1\over m^{3}},0\leq |y|\leq m,0\leq |z|\leq m\}$, such that $f_{m}|A_{2}$ satisfies conditions $(iv),(v)$ of Definition \ref{threedimensions}, where $A_{2}=\{(x,y,z):0\leq |x|\leq m+{1\over m^{3}},0\leq |y|\leq m,0\leq |z|\leq m\}$. Again, using two applications of Lemma \ref{polynomial23D14} with $n=14$ again, this time with a vertical orientation, and the fact that, for $0\leq i\leq 13$, $0\leq |x|\leq m+{1\over m^{3}}$, $0\leq |z|\leq m$, ${\partial^{i} f_{m}\over \partial y^{i}}|_{(x,m,z)}$ and ${\partial^{i} f_{m}\over \partial y^{i}}|_{(x,-m,z)}$ define smooth functions on $[-m-{1\over m^{3}},m+{1\over m^{3}}]\times [-m,m]$, we can extend $f_{m}$ to $A_{3}=\{(x,y,z):0\leq |x|\leq m+{1\over m^{3}}, m\leq |y|\leq m+{1\over m^{3}}, 0\leq |z|\leq m\}$, such that $f_{m}|A_{4}$ satisfies conditions $(vi),(vii)$ of Definition \ref{threedimensions}, where $A_{4}=\{(x,y,z):0\leq |x|\leq m+{1\over m^{3}},0\leq |y|\leq m+{1\over m^{3}},0\leq |z|\leq m\}$. Again, using two applications of Lemma \ref{polynomial23D14} with $n=14$ again, this time with a lateral orientation, and the fact that, for $0\leq i\leq 13$, $0\leq |x|\leq m+{1\over m^{3}}$, $0\leq |y|\leq m+{1\over m^{3}}$, ${\partial^{i} f_{m}\over \partial z^{i}}|_{(x,y,m)}$ and ${\partial^{i} f_{m}\over \partial z^{i}}|_{(x,y,-m)}$ define smooth functions on $[-m-{1\over m^{3}},m+{1\over m^{3}}]^{2}$, we can extend $f_{m}$ to $W_{m+{1\over m^{3}}}$ such that $f_{m}|_{W_{m+{1\over m^{3}}}}$ satisfies conditions $(viii),(ix)$ of Definition \ref{threedimensions}.

Conditions $(i),(iii)$ are then clear. We then have, using $(iii)$, that;\\

$(a)$. $\int_{\mathcal{R}^{3}}|{\partial f_{m}\over \partial x^{14}}|dxdydz=\int_{W_{m+{1\over m^{3}}}}|{\partial f_{m}\over \partial x^{14}}|dxdydz$\\

$=\int_{|x|\leq m,|y|\leq m,|z|\leq m}|{\partial f_{m}\over \partial x^{14}}|dxdydz+\int_{m\leq |x|\leq m+{1\over m^{3}},|y|\leq m,|z|\leq m}|{\partial f_{m}\over \partial x^{14}}|dxdydz$\\

$+\int_{|x|\leq m,m\leq |y|\leq m+{1\over m^{3}},|z|\leq m}|{\partial f_{m}\over \partial x^{14}}|dxdydz+\int_{m\leq |x|\leq m+{1\over m^{3}},m\leq |y|\leq m+{1\over m^{3}},|z|\leq m}|{\partial f_{m}\over \partial x^{14}}|dxdydz$\\

$+\int_{|x|\leq m,|y|\leq m,m\leq |z|\leq m+{1\over m^{3}}}|{\partial f_{m}\over \partial x^{14}}|dxdydz+\int_{m\leq |x|\leq m+{1\over m^{3}},|y|\leq m,m\leq |z|\leq m+{1\over m^{3}}}|{\partial f_{m}\over \partial x^{14}}|dxdydz$\\

$+\int_{|x|\leq m,m\leq |y|\leq m+{1\over m^{3}},m\leq |z|\leq m+{1\over m^{3}}}|{\partial f_{m}\over \partial x^{14}}|dxdydz+\int_{m\leq |x|\leq m+{1\over m^{3}},m\leq |y|\leq m+{1\over m^{3}},m\leq |z|\leq m+{1\over m^{3}}}|{\partial f_{m}\over \partial x^{14}}|dxdydz$\\

$(b)$. $\int_{\mathcal{R}^{3}}|{\partial f_{m}\over \partial y^{14}}|dxdydz=\int_{W_{m+{1\over m^{3}}}}|{\partial f_{m}\over \partial y^{14}}|dxdydz$\\

$=\int_{|x|\leq m,|y|\leq m,|z|\leq m}|{\partial f_{m}\over \partial y^{14}}|dxdydz+\int_{m\leq |x|\leq m+{1\over m^{3}},|y|\leq m,|z|\leq m}|{\partial f_{m}\over \partial y^{14}}|dxdydz$\\

$+\int_{|x|\leq m,m\leq |y|\leq m+{1\over m^{3}},|z|\leq m}|{\partial f_{m}\over \partial y^{14}}|dxdydz+\int_{m\leq |x|\leq m+{1\over m^{3}},m\leq |y|\leq m+{1\over m^{3}},|z|\leq m}|{\partial f_{m}\over \partial y^{14}}|dxdydz$\\

$+\int_{|x|\leq m,|y|\leq m,m\leq |z|\leq m+{1\over m^{3}}}|{\partial f_{m}\over \partial y^{14}}|dxdydz+\int_{m\leq |x|\leq m+{1\over m^{3}},|y|\leq m,m\leq |z|\leq m+{1\over m^{3}}}|{\partial f_{m}\over \partial y^{14}}|dxdydz$\\

$+\int_{|x|\leq m,m\leq |y|\leq m+{1\over m^{3}},m\leq |z|\leq m+{1\over m^{3}}}|{\partial f_{m}\over \partial y^{14}}|dxdydz+\int_{m\leq |x|\leq m+{1\over m^{3}},m\leq |y|\leq m+{1\over m^{3}},m\leq |z|\leq m+{1\over m^{3}}}|{\partial f_{m}\over \partial y^{14}}|dxdydz$\\

$(c)$. $\int_{\mathcal{R}^{3}}|{\partial f_{m}\over \partial z^{14}}|dxdydz=\int_{W_{m+{1\over m^{3}}}}|{\partial f_{m}\over \partial z^{14}}|dxdydz$\\

$=\int_{|x|\leq m,|y|\leq m,|z|\leq m}|{\partial f_{m}\over \partial z^{14}}|dxdydz+\int_{m\leq |x|\leq m+{1\over m^{3}},|y|\leq m,|z|\leq m}|{\partial f_{m}\over \partial z^{14}}|dxdydz$\\

$+\int_{|x|\leq m,m\leq |y|\leq m+{1\over m^{3}},|z|\leq m}|{\partial f_{m}\over \partial z^{14}}|dxdydz+\int_{m\leq |x|\leq m+{1\over m^{3}},m\leq |y|\leq m+{1\over m^{3}},|z|\leq m}|{\partial f_{m}\over \partial z^{14}}|dxdydz$\\

$+\int_{|x|\leq m,|y|\leq m,m\leq |z|\leq m+{1\over m^{3}}}|{\partial f_{m}\over \partial z^{14}}|dxdydz+\int_{m\leq |x|\leq m+{1\over m^{3}},|y|\leq m,m\leq |z|\leq m+{1\over m^{3}}}|{\partial f_{m}\over \partial z^{14}}|dxdydz$\\

$+\int_{|x|\leq m,m\leq |y|\leq m+{1\over m^{3}},m\leq |z|\leq m+{1\over m^{3}}}|{\partial f_{m}\over \partial z^{14}}|dxdydz+\int_{m\leq |x|\leq m+{1\over m^{3}},m\leq |y|\leq m+{1\over m^{3}},m\leq |z|\leq m+{1\over m^{3}}}|{\partial f_{m}\over \partial z^{14}}|dxdydz$\\

$(*)$\\

We then have the following cases, using the second clause in Lemma \ref{polynomial23D14} repeatedly with the appropriate orientations;\\

Case 1;\\

$\int_{|x|\leq m,|y|\leq m,|z|\leq m}|{\partial^{14}f_{m}\over \partial x^{14}}|dxdydz$\\

$=\int_{|x|\leq m,|y|\leq m,|z|\leq m}|{\partial^{14}f\over \partial x^{14}}|dxdydz\leq Fm^{3}$\\

$\int_{|x|\leq m,|y|\leq m,|z|\leq m}|{\partial^{14}f_{m}\over \partial y^{14}}|dxdydz$\\

$=\int_{|x|\leq m,|y|\leq m,|z|\leq m}|{\partial^{14}f\over \partial y^{14}}|dxdydz\leq Fm^{3}$\\

$\int_{|x|\leq m,|y|\leq m,|z|\leq m}|{\partial^{14}f_{m}\over \partial z^{14}}|dxdydz$\\

$=\int_{|x|\leq m,|y|\leq m,|z|\leq m}|{\partial^{14}f\over \partial z^{14}}|dxdydz\leq Fm^{3}$\\

Case 2;\\

$\int_{m\leq |x|\leq m+{1\over m^{3}},|y|\leq m,|z|\leq m}|{\partial^{14}f_{m}\over \partial x^{14}}|dxdydz$\\

$=\int_{|y|\leq m,|z|\leq m}(\int_{m\leq |x|\leq m+{1\over m^{3}}}|{\partial^{14}f_{m}\over \partial x^{14}}|dx)dydz$\\

$=\int_{|y|\leq m,|z|\leq m}(|{\partial^{13}f\over \partial x^{13}}|_{(m,y,z)}+|{\partial^{13}f\over \partial x^{13}}|_{(-m,y,z)})dydz$\\

$\leq 2(2m)^{2}F$\\

$=8m^{2}F$\\

Case 3;\\

$\int_{m\leq |x|\leq m+{1\over m^{3}},|y|\leq m,|z|\leq m}|{\partial^{14}f_{m}\over \partial y^{14}}|dxdydz$\\

$=\int_{|y|\leq m,|z|\leq m}(\int_{m\leq |x|\leq m+{1\over m^{3}}}|{\partial^{14}f_{m}\over \partial y^{14}}|dx)dydz$\\

$\leq {1\over m^{3}}\int_{|y|\leq m,|z|\leq m}(|\sum_{i=0}^{13}D_{i}|{\partial^{i}\partial^{14}f\over \partial y^{14}\partial x^{i}}|(m,y,z)+|\sum_{i=0}^{13}D_{i}|{\partial^{i}\partial^{14}f\over \partial y^{14}\partial x^{i}}|(-m,y,z))dydz$\\

$\leq {2\over m^{3}}(2m)^{2}F(\sum_{i=0}^{13}D_{i})$\\

$={8F(\sum_{i=0}^{13}D_{i})\over m}$\\

Case 4;\\

$\int_{m\leq |x|\leq m+{1\over m^{3}},|y|\leq m,|z|\leq m}|{\partial^{14}f_{m}\over \partial z^{14}}|dxdydz$\\

$=\int_{|y|\leq m,|z|\leq m}(\int_{m\leq |x|\leq m+{1\over m^{3}}}|{\partial^{14}f_{m}\over \partial z^{14}}|dx)dydz$\\

$\leq {1\over m^{3}}\int_{|y|\leq m,|z|\leq m}(|\sum_{i=0}^{13}D_{i}|{\partial^{i}\partial^{14}f\over \partial z^{14}\partial x^{i}}|(m,y,z)+|\sum_{i=0}^{13}D_{i}|{\partial^{i}\partial^{14}f\over \partial z^{14}\partial x^{i}}|(-m,y,z))dydz$\\

$\leq {2\over m^{3}}(2m)^{2}F(\sum_{i=0}^{13}D_{i})$\\

$={8F(\sum_{i=0}^{13}D_{i})\over m}$\\

Case 5.\\

$\int_{|x|\leq m,m\leq |y|\leq m+{1\over m^{3}},|z|\leq m}|{\partial^{14}f_{m}\over \partial x^{14}}|dxdydz$\\

$=\int_{|x|\leq m,|z|\leq m}(\int_{|y|\leq m+{1\over m^{3}}}|{\partial^{14}f_{m}\over \partial x^{14}}|dy)dxdz$\\

$\leq {2\over m^{3}}\int_{|x|\leq m,|z|\leq m}C_{14}dx$\\

$=(2m)^{2}{2\over m^{3}}C_{14,0}$\\

$={8C_{14,0}\over m}$\\

Csse 6.\\

$\int_{|x|\leq m,m\leq |y|\leq m+{1\over m^{3}},|z|\leq m}|{\partial^{14}f_{m}\over \partial z^{14}}|dxdydz$\\

$=\int_{|x|\leq m,|z|\leq m}(\int_{|y|\leq m+{1\over m^{3}}}|{\partial^{14}f_{m}\over \partial z^{14}}|dy)dxdz$\\

$\leq {2\over m^{3}}\int_{|x|\leq m,|z|\leq m}C_{0,14}dx$\\

$=(2m)^{2}{2\over m^{3}}C_{0,14}$\\

$={8C_{0,14}\over m}$\\

Case 7.\\

$\int_{|x|\leq m,m\leq |y|\leq m+{1\over m^{3}},|z|\leq m}|{\partial^{14}f_{m}\over \partial y^{14}}|dxdydz$\\

$= \int_{|x|\leq m,|z|\leq m}(\int_{m\leq |y|\leq m+{1\over m^{3}}}|{\partial^{14}f_{m}\over \partial y^{14}}|dy)dxdz$\\

$= \int_{|x|\leq m,|z|\leq m}(|{\partial f\over \partial y^{13}}|_{(x,m,z)}+|{\partial f\over \partial y^{13}}|_{(x,-m,z)})dxdz)$\\

$\leq 2(2m)^{2}F$\\

$=8m^{2}F$\\

Case 8.\\

$\int_{m\leq |x|\leq m+{1\over m^{3}},m\leq |y|\leq m+{1\over m^{3}},|z|\leq m}|{\partial^{14}f_{m}\over \partial x^{14}}|dxdydz$\\

$=\int_{m\leq |x|\leq m+{1\over m^{3}},|z|\leq m}(\int_{m\leq |y|\leq m+{1\over m^{3}}}|{\partial^{14}f_{m}\over \partial x^{14}}|dy)dxdz$\\

$\leq {1\over m^{3}}\int_{m\leq |x|\leq m+{1\over m^{3}},|z|\leq m}(\sum_{i=0}^{13}L_{i,14}|{\partial^{i+14}\partial^{14}f_{m}\over \partial y^{i}\partial x^{14}}|_{(x,m,z)}+L_{i,14}|{\partial^{i+14}\partial^{14}f_{m}\over \partial y^{i}\partial x^{14}}|_{(x,-m,z)})dxdz$\\

$={1\over m^{3}}\int_{|z|\leq m}(\sum_{i=0}^{13}L_{i,14}(|{\partial^{i+13}\partial^{14}f\over \partial y^{i}\partial x^{13}}|_{(m,m,z)}|+|{\partial^{i+13}\partial^{14}f\over \partial y^{i}\partial x^{13}}|_{(m,-m,z)}|+|{\partial^{i+13}\partial^{14}f\over \partial y^{i}\partial x^{13}}|_{(-m,m,z)}|$\\

$+|{\partial^{i+13}\partial^{14}f\over \partial y^{i}\partial x^{13}}|_{(-m,-m,z)}|))dz$\\

$\leq (2m){4F(\sum_{i=0}^{13}L_{i,14})\over m^{3}}$\\

$={8F(\sum_{i=0}^{13}L_{i,14})\over m^{2}}$\\

(the constants $L_{i,14}$,$0\leq i\leq 13$ coming from the proof of Lemma \ref{polynomial22D14})\\

Case 9.\\

$\int_{m\leq |x|\leq m+{1\over m^{3}},m\leq |y|\leq m+{1\over m^{3}},|z|\leq m}|{\partial^{14}f_{m}\over \partial y^{14}}|dxdydz$\\

$=\int_{m\leq |x|\leq m+{1\over m^{3}},|z|\leq m}(\int_{m\leq |y|\leq m+{1\over m^{3}}}|{\partial^{14}f_{m}\over \partial y^{14}}|dy)dxdz$\\

$=\int_{m\leq |x|\leq m+{1\over m^{3}},|z|\leq m}(|{\partial^{13}f_{m}\over \partial y^{13}}|_{(x,m,z)}+|{\partial^{13}f_{m}\over \partial y^{13}}|_{(x,-m,z)})dxdz$\\

$\leq {1\over m^{3}}(\int_{|z|\leq m}C_{13,1}dz+\int_{|z|\leq m}C_{13,2}dz)$\\

$\leq (2m){max(C_{13,1},C_{13,2})\over m^{3}}$\\

$={2max(C_{13,1},C_{13,2})\over m^{2}}$\\

 (the constants $\{C_{13,1},C_{13,2}\}$ coming from the two applications of Lemma \ref{polynomial22D14} at the two boundaries)\\

Case 10.\\

$\int_{m\leq |x|\leq m+{1\over m^{3}},m\leq |y|\leq m+{1\over m^{3}},|z|\leq m}|{\partial^{14}f_{m}\over \partial z^{14}}|dxdydz$\\

$=\int_{m\leq |x|\leq m+{1\over m^{3}},|z|\leq m}(\int_{m\leq |y|\leq m+{1\over m^{3}}}|{\partial^{14}f_{m}\over \partial z^{14}}|dy)dxdz$\\

$\leq {1\over m^{3}}\int_{m\leq |x|\leq m+{1\over m^{3}},|z|\leq m}(\sum_{i=0}^{13}L_{i,14}|{\partial^{i+14}f_{m}\over \partial y^{i}\partial z^{14}}|_{(x,m,z)}+L_{i,14}|{\partial^{i+14}f_{m}\over \partial y^{i}\partial z^{14}}|_{(x,-m,z)})dxdz$\\

$\leq {1\over m^{6}}\int_{|z|\leq m}(\sum_{i=0}^{13}\sum_{j=0}^{13}L_{i,14}L_{j,i,14}(|{\partial^{i+j+14}f\over \partial x^{j}\partial y^{i}\partial z^{14}}|_{(m,m,z)}|+|{\partial^{i+j+14}f\over \partial x^{j}\partial y^{i}\partial z^{14}}|_{(m,-m,z)}|$\\

$+|{\partial^{i+j+14}f\over \partial x^{j}\partial y^{i}\partial z^{14}}|_{(-m,m,z)}|+|{\partial^{i+j+14}f\over \partial x^{j}\partial y^{i}\partial z^{14}}|_{(-m,-m,z)}|))dz$\\

$\leq (2m){4F(\sum_{i=0}^{13}\sum_{j=0}^{13}L_{i,14}L_{j,i,14})\over m^{6}}$\\

$={8F(\sum_{i=0}^{13}\sum_{j=0}^{13}L_{i,14}L_{j,i,14})\over m^{5}}$\\

(the constants $L_{i,14},L_{j,i,14}$,$0\leq i\leq 13,0\leq j\leq 13$ coming from two applications of the proof of Lemma \ref{polynomial23D14})\\

Case 11.\\

$\int_{|x|\leq m,|y|\leq m,m\leq |z|\leq m+{1\over m^{3}}}|{\partial f_{m}\over \partial x^{14}}|dxdydz$\\

$=\int_{|x|\leq m,|y|\leq m}(\int_{m\leq |z|\leq m+{1\over m^{3}}}|{\partial f_{m}\over \partial x^{14}}|dz)dxdy$\\

$\leq {2\over m^{3}}\int_{|x|\leq m,|y|\leq m}(E_{14,0})$\\

$=(2m)^{2}{2\over m^{3}}E_{14,0}$\\

$={8E_{14,0}\over m}$\\

Case 12.\\

$\int_{|x|\leq m,|y|\leq m,m\leq |z|\leq m+{1\over m^{3}}}|{\partial f_{m}\over \partial y^{14}}|dxdydz$\\

$=\int_{|x|\leq m,|y|\leq m}(\int_{m\leq |z|\leq m+{1\over m^{3}}}|{\partial f_{m}\over \partial y^{14}}|dz)dxdy$\\

$\leq {2\over m^{3}}\int_{|x|\leq m,|y|\leq m}(E_{0,14})$\\

$=(2m)^{2}{2\over m^{3}}E_{0,14}$\\

$={8E_{0,14}\over m}$\\

(the constants $E_{0,14},E_{14,0}$ coming from an application of Lemma \ref{polynomial23D14} with a different orientation)\\

Case 13.

$\int_{|x|\leq m,|y|\leq m,m\leq |z|\leq m+{1\over m^{3}}}|{\partial f_{m}\over \partial z^{14}}|dxdydz$\\

$=\int_{|x|\leq m,|y|\leq m}(\int_{m\leq |z|\leq m+{1\over m^{3}}}|{\partial f_{m}\over \partial z^{14}}|dz)dxdy$\\

$=\int_{|x|\leq m,|y|\leq m}(|{\partial f\over \partial z^{13}}|(x,y,m)+|{\partial f\over \partial z^{13}}|(x,y,m))dxdy$\\

$\leq 2(2m)^{2}F$\\

$=8m^{2}F$\\

Case 14.\\

$\int_{m\leq |x|\leq m+{1\over m^{3}},|y|\leq m,m\leq |z|\leq m+{1\over m^{3}}}|{\partial^{14} f_{m}\over \partial x^{14}}|dxdydz$\\

$=\int_{m\leq |x|\leq m+{1\over m^{3}},|y|\leq m}(\int_{m\leq |z|\leq m+{1\over m^{3}}}|{\partial^{14} f_{m}\over \partial x^{14}}|dz)dxdy$\\

$\leq {1\over m^{3}}\int_{m\leq |x|\leq m+{1\over m^{3}},|y|\leq m}(\sum_{i=0}^{13}L_{i,14}|{\partial^{i+14} f_{m}\over \partial z^{i}\partial x^{14}}|(x,y,m)+L_{i,14}|{\partial^{i+14} f_{m}\over \partial z^{i}\partial x^{14}}|(x,y,-m))dxdy$\\

$={1\over m^{3}}\int_{|y|\leq m}(\int_{m\leq |x|\leq m+{1\over m^{3}}}(\sum_{i=0}^{13}L_{i,14}(|{\partial^{i+14} f_{m}\over \partial z^{i}\partial x^{14}}|(x,y,m)+L_{i,14}(|{\partial^{i+14} f_{m}\over \partial z^{i}\partial x^{14}}|(x,y,-m))dx)dy$\\

$={1\over m^{3}}\int_{|y|\leq m}(\sum_{i=0}^{13}L_{i,14}|{\partial^{i+13} f\over \partial z^{i}\partial x^{13}}|(m,y,m)+\sum_{i=0}^{13}L_{i,14}|{\partial^{i+13} f\over \partial z^{i}\partial x^{13}}|(-m,y,m)$\\

$+\sum_{i=0}^{13}L_{i,14}|{\partial^{i+13} f\over \partial z^{i}\partial x^{13}}|(m,y,m)+\sum_{i=0}^{13}L_{i,14}|{\partial^{i+13} f\over \partial z^{i}\partial x^{13}}|(-m,y,-m))dy$\\

$\leq (2m){1\over m^{3}}(4F)(\sum_{i=0}^{13}L_{i,14})$\\

$={8F(\sum_{i=0}^{13}L_{i,14})\over m^{2}}$\\

Case 15.\\

$\int_{m\leq |x|\leq m+{1\over m^{3}},|y|\leq m,m\leq |z|\leq m+{1\over m^{3}}}|{\partial^{14} f_{m}\over \partial y^{14}}|dxdydz$\\

$=\int_{m\leq |x|\leq m+{1\over m},|y|\leq m}(\int_{m\leq |z|\leq m+{1\over m^{3}}}|{\partial^{14} f_{m}\over \partial y^{14}}|dz)dxdy$\\

$\leq {1\over m^{3}}\int_{m\leq |x|\leq m+{1\over m^{3}},|y|\leq m}(\sum_{i=0}^{13}L_{i,14}|{\partial^{i+14} f_{m}\over \partial z^{i}\partial y^{14}}|(x,y,m)+L_{i,14}|{\partial^{i+14} f_{m}\over \partial y^{i}\partial x^{14}}|(x,y,-m))dxdy$\\

$={1\over m^{3}}\int_{|y|\leq m}(\int_{m\leq |x|\leq m+{1\over m^{3}}}(\sum_{i=0}^{13}L_{i,14}(|{\partial^{i+14} f_{m}\over \partial z^{i}\partial y^{14}}|(x,y,m)+L_{i,14}(|{\partial^{i+14} f_{m}\over \partial z^{i}\partial y^{14}}|(x,y,-m))dx)dy$\\

$\leq {1\over m^{6}}\int_{|y|\leq m}(\sum_{i=0}^{13}\sum_{j=0}^{13}L_{i,14}L_{i,j,14}|{\partial^{i+j+14} f\over \partial x^{j}\partial z^{i}\partial y^{14}}|(m,y,m)$\\

$+\sum_{i=0}^{13}\sum_{j=0}^{13}L_{i,14}L_{i,j,14}|{\partial^{i+j+14} f\over \partial x^{j}\partial z^{i}\partial y^{14}}|(-m,y,m)$\\

$+\sum_{i=0}^{13}\sum_{j=0}^{13}L_{i,14}L_{i,j,14}|{\partial^{i+j+14} f\over \partial x^{j}\partial z^{i}\partial y^{14}}|(m,y,-m)$\\

$+\sum_{i=0}^{13}\sum_{j=0}^{13}L_{i,14}L_{i,j,14}|{\partial^{i+j+14} f\over \partial x^{j}\partial z^{i}\partial y^{14}}|(-m,y,-m))dy$\\

$\leq (2m){1\over m^{6}}(4F)(\sum_{i=0}^{13}\sum_{j=0}^{13}L_{i,14}L_{i,j,14})$\\

$={8F(\sum_{i=0}^{13}\sum_{j=0}^{13}L_{i,14}L_{i,j,14})\over m^{5}}$\\

Case 16.\\

$\int_{m\leq |x|\leq m+{1\over m^{3}},|y|\leq m,m\leq |z|\leq m+{1\over m^{3}}}|{\partial^{14} f_{m}\over \partial z^{14}}|dxdydz$\\

$=\int_{m\leq |x|\leq m+{1\over m^{3}},|y|\leq m}(\int_{m\leq |z|\leq m+{1\over m^{3}}}|{\partial^{14} f_{m}\over \partial z^{14}}|dz)dxdy$\\

$=\int_{m\leq |x|\leq m+{1\over m^{3}},|y|\leq m}(|{\partial^{13} f_{m}\over \partial z^{13}}|(x,y,m)+|{\partial^{13} f_{m}\over \partial z^{13}}|(x,y,-m))dxdy$\\

$=\int_{|y|\leq m}(\int_{m\leq |x|\leq m+{1\over m^{3}}}(|{\partial^{13} f_{m}\over \partial z^{13}}|(x,y,m)+|{\partial^{13} f_{m}\over \partial z^{13}}|(x,y,-m))dx)dy$\\

$\leq {1\over m^{3}}\int_{|y|\leq m}(\sum_{i=0}^{13}L_{i,13}|{\partial^{i+13} f\over \partial x^{i}\partial z^{13}}|(m,y,m)+\sum_{i=0}^{13}L_{i,13}|{\partial^{i+13} f\over \partial x^{i}\partial z^{13}}|(-m,y,m)$\\

$+\sum_{i=0}^{13}L_{i,13}|{\partial^{i+13} f\over \partial x^{i}\partial z^{13}}|(m,y,-m)+\sum_{i=0}^{13}L_{i,13}|{\partial^{i+13} f\over \partial x^{i}\partial z^{13}}|(-m,y,-m))$\\

$\leq (2m){1\over m^{3}}(4F)(\sum_{i=0}^{13}L_{i,13})$\\

$={8F(\sum_{i=0}^{13}L_{i,13})\over m^{2}}$\\

Cases 17-19 are similar to cases 14-16, interchanging the orders of integration, with case 17 corresponding to case 15, case 18 corresponding to case 14 and case 19 corresponding to case 16, so that;\\

Case 17.\\

$\int_{|x|\leq m,m\leq |y|\leq m+{1\over m^{3}},m\leq |z|\leq m+{1\over m^{3}}}|{\partial^{14} f_{m}\over \partial x^{14}}|dxdydz$\\

$\leq {8F(\sum_{i=0}^{13}\sum_{j=0}^{13}L_{i,14}L_{i,j,14})\over m^{5}}$\\

Case 18.\\

$\int_{|x|\leq m,m\leq |y|\leq m+{1\over m^{3}},m\leq |z|\leq m+{1\over m^{3}}}|{\partial^{14} f_{m}\over \partial y^{14}}|dxdydz$\\

$\leq {8F(\sum_{i=0}^{13}L_{i,14})\over m^{2}}$\\

Case 19.\\

$\int_{|x|\leq m,m\leq |y|\leq m+{1\over m^{3}},m\leq |z|\leq m+{1\over m^{3}}}|{\partial^{14} f_{m}\over \partial z^{14}}|dxdydz$\\

$\leq {8F(\sum_{i=0}^{13}L_{i,13})\over m^{2}}$\\

Case 20.\\

$\int_{m\leq |x|\leq m+{1\over m^{3}},m\leq |y|\leq m+{1\over m^{3}},m\leq |z|\leq m+{1\over m^{3}}}|{\partial^{14} f_{m}\over \partial x^{14}}|dxdydz$\\

$=\int_{m\leq |x|\leq m+{1\over m^{3}},m\leq |y|\leq m+{1\over m^{3}}}(\int_{m\leq |z|\leq m+{1\over m^{3}}}|{\partial^{14} f_{m}\over \partial x^{14}}|dz)dxdy$\\

$\leq {1\over m^{3}}\int_{m\leq |x|\leq m+{1\over m^{3}},m\leq |y|\leq m+{1\over m^{3}}}(\sum_{i=0}^{13}L_{i,14}|{\partial^{i+14} f_{m}\over \partial z^{i}\partial x^{14}}|(x,y,m)+\sum_{i=0}^{13}L_{i,14}|{\partial^{i+14} f_{m}\over \partial z^{i}\partial x^{14}}|(x,y,-m))dxdy$\\

$={1\over m^{3}}\int_{m\leq |x|\leq m+{1\over m^{3}}}(\int_{m\leq |y|\leq m+{1\over m^{3}}}(\sum_{i=0}^{13}L_{i,14}(|{\partial^{i+14} f_{m}\over \partial z^{i}\partial x^{14}}|(x,y,m)$\\

$+\sum_{i=0}^{13}L_{i,14}(|{\partial^{i+14} f_{m}\over \partial z^{i}\partial x^{14}}|(x,y,-m))dy)dx$\\

$\leq {1\over m^{6}}\int_{m\leq |x|\leq m+{1\over m^{3}}}(\sum_{j=0}^{13}\sum_{i=0}^{13}L_{i,14}L_{i,j,14}|{\partial^{i+j+14} f_{m}\over \partial y^{j}\partial z^{i}\partial x^{14}}|(x,m,m)$\\

$+\sum_{j=0}^{13}\sum_{i=0}^{13}L_{i,14}L_{i,j,14}|{\partial^{i+j+14} f_{m}\over \partial y^{j}\partial z^{i}\partial x^{14}}|(x,-m,m)$\\

$+\sum_{j=0}^{13}\sum_{i=0}^{13}L_{i,14}L_{i,j,14}|{\partial^{i+j+14} f_{m}\over \partial y^{j}\partial z^{i}\partial x^{14}}|(x,m,-m)$\\

$+\sum_{j=0}^{13}\sum_{i=0}^{13}L_{i,14}L_{i,j,14}|{\partial^{i+j+14} f_{m}\over \partial y^{j}\partial z^{i}\partial x^{14}}|(x,-m,-m))dx$\\

$={1\over m^{6}}(\sum_{j=0}^{13}\sum_{i=0}^{13}L_{i,14}L_{i,j,14}|{\partial^{i+j+13} f\over \partial y^{j}\partial z^{i}\partial x^{13}}|(m,m,m)+\sum_{j=0}^{13}\sum_{i=0}^{13}L_{i,14}L_{i,j,14}|{\partial^{i+j+13} f\over \partial y^{j}\partial z^{i}\partial x^{13}}|(-m,m,m)$\\

$+\sum_{j=0}^{13}\sum_{i=0}^{13}L_{i,14}L_{i,j,14}|{\partial^{i+j+13} f\over \partial y^{j}\partial z^{i}\partial x^{13}}|(m,-m,m)+\sum_{j=0}^{13}\sum_{i=0}^{13}L_{i,14}L_{i,j,14}|{\partial^{i+j+13} f\over \partial y^{j}\partial z^{i}\partial x^{13}}|(-m,-m,m)$\\

$+\sum_{j=0}^{13}\sum_{i=0}^{13}L_{i,14}L_{i,j,14}|{\partial^{i+j+13} f\over \partial y^{j}\partial z^{i}\partial x^{13}}|(m,m,-m)+\sum_{j=0}^{13}\sum_{i=0}^{13}L_{i,14}L_{i,j,14}|{\partial^{i+j+13} f\over \partial y^{j}\partial z^{i}\partial x^{13}}|(-m,m,-m)$\\

$+\sum_{j=0}^{13}\sum_{i=0}^{13}L_{i,14}L_{i,j,14}|{\partial^{i+j+13} f\over \partial y^{j}\partial z^{i}\partial x^{13}}|(m,-m,-m)$\\

$+\sum_{j=0}^{13}\sum_{i=0}^{13}L_{i,14}L_{i,j,14}|{\partial^{i+j+13} f\over \partial y^{j}\partial z^{i}\partial x^{13}}|(-m,-m,-m))$\\

$\leq {8F\over m^{6}}(\sum_{j=0}^{13}\sum_{i=0}^{13}L_{i,14}L_{i,j,14})$\\

Case 21.\\

$\int_{m\leq |x|\leq m+{1\over m^{3}},m\leq |y|\leq m+{1\over m^{3}},m\leq |z|\leq m+{1\over m^{3}}}|{\partial^{14} f_{m}\over \partial y^{14}}|dxdydz$\\

$=\int_{m\leq |x|\leq m+{1\over m^{3}},m\leq |y|\leq m+{1\over m^{3}}}(\int_{m\leq |z|\leq m+{1\over m^{3}}}|{\partial^{14} f_{m}\over \partial y^{14}}|dz)dxdy$\\

$\leq {1\over m^{3}}\int_{m\leq |x|\leq m+{1\over m^{3}},m\leq |y|\leq m+{1\over m^{3}}}(\sum_{i=0}^{13}L_{i,14}|{\partial^{i+14} f_{m}\over \partial z^{i}\partial y^{14}}|(x,y,m)+\sum_{i=0}^{13}L_{i,14}|{\partial^{i+14} f_{m}\over \partial z^{i}\partial y^{14}}|(x,y,-m))dxdy$\\

$={1\over m^{3}}\int_{|x|\leq m+{1\over m^{3}}}(\int_{m\leq |y|\leq m+{1\over m^{3}}}(\sum_{i=0}^{13}L_{i,14}(|{\partial^{i+14} f_{m}\over \partial z^{i}\partial y^{14}}|(x,y,m)+\sum_{i=0}^{13}L_{i,14}(|{\partial^{i+14} f_{m}\over \partial z^{i}\partial y^{14}}|(x,y,-m))dy)dx$\\

$={1\over m^{3}}\int_{m\leq |x|\leq m+{1\over m^{3}}}(\sum_{i=0}^{13}L_{i,14}|{\partial^{i+13} f_{m}\over \partial z^{i}\partial y^{13}}|(x,m,m)$\\

$+\sum_{i=0}^{13}L_{i,14}|{\partial^{i+13} f_{m}\over \partial z^{i}\partial y^{13}}|(x,-m,m)$\\

$+\sum_{i=0}^{13}L_{i,14}|{\partial^{i+13} f_{m}\over \partial z^{i}\partial y^{13}}|(x,m,-m)$\\

$+\sum_{i=0}^{13}L_{i,14}|{\partial^{i+13} f_{m}\over \partial z^{i}\partial y^{13}}|(x,-m,-m))dx$\\

$\leq {1\over m^{6}}(\sum_{j=0}^{13}\sum_{i=0}^{13}L_{i,14}L_{i,j,13}|{\partial^{i+j+13} f\over \partial x^{j}\partial z^{i}\partial y^{13}}|(m,m,m)+\sum_{j=0}^{13}\sum_{i=0}^{13}L_{i,14}L_{i,j,13}|{\partial^{i+j+13} f\over \partial x^{j}\partial z^{i}\partial y^{13}}|(-m,m,m)$\\

$+\sum_{j=0}^{13}\sum_{i=0}^{13}L_{i,14}L_{i,j,13}|{\partial^{i+j+13} f\over \partial x^{j}\partial z^{i}\partial y^{13}}|(m,-m,m)+\sum_{j=0}^{13}\sum_{i=0}^{13}L_{i,14}L_{i,j,13}|{\partial^{i+j+13} f\over \partial x^{j}\partial z^{i}\partial y^{13}}|(-m,-m,m)$\\

$+\sum_{j=0}^{13}\sum_{i=0}^{13}L_{i,14}L_{i,j,13}|{\partial^{i+j+13} f\over \partial x^{j}\partial z^{i}\partial y^{13}}|(m,m,-m)+\sum_{j=0}^{13}\sum_{i=0}^{13}L_{i,14}L_{i,j,13}|{\partial^{i+j+13} f\over \partial x^{j}\partial z^{i}\partial y^{13}}|(-m,m,-m)$\\

$+\sum_{j=0}^{13}\sum_{i=0}^{13}L_{i,14}L_{i,j,13}|{\partial^{i+j+13} f\over \partial x^{j}\partial z^{i}\partial y^{13}}|(m,-m,-m)$\\

$+\sum_{j=0}^{13}\sum_{i=0}^{13}L_{i,14}L_{i,j,13}|{\partial^{i+j+13} f\over \partial x^{j}\partial z^{i}\partial y^{13}}|(-m,-m,-m))$\\

$\leq {8F(\sum_{j=0}^{13}\sum_{i=0}^{13}L_{i,14}L_{i,j,13})\over m^{6}}$\\

Case 22.\\

$\int_{m\leq |x|\leq m+{1\over m^{3}},m\leq |y|\leq m+{1\over m^{3}},m\leq |z|\leq m+{1\over m^{3}}}|{\partial^{14} f_{m}\over \partial z^{14}}|dxdydz$\\

$=\int_{m\leq |x|\leq m+{1\over m^{3}},m\leq |y|\leq m+{1\over m^{3}}}(\int_{m\leq |z|\leq m+{1\over m^{3}}}|{\partial^{14} f_{m}\over \partial z^{14}}|dz)dxdy$\\

$=\int_{m\leq |x|\leq m+{1\over m^{3}},m\leq |y|\leq m+{1\over m^{3}}}(|{\partial^{13}f_{m}\over \partial z^{13}}|(x,y,m)+|{\partial^{13}f_{m}\over \partial z^{13}}|(x,y,-m))dxdy$\\

$=\int_{m\leq |x|\leq m+{1\over m^{3}}}(\int_{m\leq |y|\leq m+{1\over m^{3}}}((|{\partial^{13}f_{m}\over \partial z^{13}}|(x,y,m)+|{\partial^{13}f_{m}\over \partial z^{13}}|(x,y,-m))dy)dx$\\

$\leq {1\over m^{3}}\int_{|x|\leq m+{1\over m^{3}}}(\sum_{i=0}^{13}L_{i,13}|{\partial^{i+13} f_{m}\over \partial y^{i}\partial z^{13}}|(x,m,m)$\\

$+\sum_{i=0}^{13}L_{i,13}|{\partial^{i+13} f_{m}\over \partial y^{i}\partial z^{13}}|(x,-m,m)$\\

$+\sum_{i=0}^{13}L_{i,13}|{\partial^{i+13} f_{m}\over \partial y^{i}\partial z^{13}}|(x,m,-m)$\\

$+\sum_{i=0}^{13}L_{i,13}|{\partial^{i+13} f_{m}\over \partial y^{i}\partial z^{13}}|(x,-m,-m))dx$\\

$\leq {1\over m^{6}}(\sum_{j=0}^{13}\sum_{i=0}^{13}L_{i,13}L_{i,j,13}|{\partial^{i+j+13} f\over \partial x^{j}\partial y^{i}\partial z^{13}}|(m,m,m)+\sum_{j=0}^{13}\sum_{i=0}^{13}L_{i,13}L_{i,j,13}|{\partial^{i+j+13} f\over \partial x^{j}\partial y^{i}\partial z^{13}}|(-m,m,m)$\\

$+\sum_{j=0}^{13}\sum_{i=0}^{13}L_{i,13}L_{i,j,13}|{\partial^{i+j+13} f\over \partial x^{j}\partial y^{i}\partial z^{13}}|(m,-m,m)+\sum_{j=0}^{13}\sum_{i=0}^{13}L_{i,13}L_{i,j,13}|{\partial^{i+j+13} f\over \partial x^{j}\partial y^{i}\partial z^{13}}|(-m,-m,m)$\\

$+\sum_{j=0}^{13}\sum_{i=0}^{13}L_{i,13}L_{i,j,13}|{\partial^{i+j+13} f\over \partial x^{j}\partial y^{i}\partial z^{13}}|(m,m,-m)+\sum_{j=0}^{13}\sum_{i=0}^{13}L_{i,13}L_{i,j,13}|{\partial^{i+j+13} f\over \partial x^{j}\partial y^{i}\partial z^{13}}|(-m,m,-m)$\\

$+\sum_{j=0}^{13}\sum_{i=0}^{13}L_{i,13}L_{i,j,13}|{\partial^{i+j+13} f\over \partial x^{j}\partial y^{i}\partial z^{13}}|(m,-m,-m)$\\

$+\sum_{j=0}^{13}\sum_{i=0}^{13}L_{i,13}L_{i,j,13}|{\partial^{i+j+13} f\over \partial x^{j}\partial y^{i}\partial z^{13}}|(-m,-m,-m))$\\

$\leq {8F(\sum_{j=0}^{13}\sum_{i=0}^{13}L_{i,13}L_{i,j,13})\over m^{6}}$\\

It is then clear from $(*)$, summing the bounds from the individual cases 1-19, as at the end of the proof of Lemma \ref{existence}, that there exists a constant $G\in\mathcal{R}_{>0}$ with;\\

$max(\int_{\mathcal{R}^{3}}|{\partial f_{m}\over \partial x^{14}}|dxdydz,\int_{\mathcal{R}^{3}}|{\partial f_{m}\over \partial y^{14}}|dxdydz,\int_{\mathcal{R}^{3}}|{\partial f_{m}\over \partial z^{14}}|dxdydz)\leq Gm^{3}$\\

for sufficiently large $m$.\\

\end{proof}
\begin{lemma}
\label{decayofapprox}
Let $\{f_{m}:m\in\mathcal{N}\}$ be the inflexionary sequence constructed in Lemma \ref{existence2}, then for $\overline{k}\in\mathcal{R}^{3}$, $\overline{k}\neq \overline{0}$, sufficiently large $m$, we have that there exists $D\in\mathcal{R}_{>0}$, independent of $m$, with;\\

$|\mathcal{F}(f_{m})(\overline{k})|\leq {Dm^{3}\over |\overline{k}|^{14}}$\\

Moreover, for sufficiently large $m$, $\mathcal{F}(f_{m})\in L^{1}(\mathcal{R}^{3})$.\\

A similar result holds for the inflexionary sequence $\{f_{m}:m\in\mathcal{N}\}$, constructed in Lemma \ref{existence}, for $\overline{k}\neq 0$, sufficiently large $m$, we have that there exists $D\in\mathcal{R}_{>0}$, independent of $m$, with;\\

$|\mathcal{F}(f_{m})(\overline{k})|\leq {Dm^{2}\over |\overline{k}|^{14}}$\\

Moreover, for sufficiently large $m$, $\mathcal{F}(f_{m})\in L^{1}(\mathcal{R}^{3})$.\\

\end{lemma}

\begin{proof}

For $(k_{1},k_{2},k_{3})\in\mathcal{R}^{3}$, using repeated integration by parts, and the fact that;\\

$\{{\partial f_{m}\over \partial x^{14}},{\partial f_{m}\over \partial y^{14}},{\partial f_{m}\over \partial z^{14}}\}\subset L^{1}(\mathcal{R}^{3})$\\

$\{{\partial f_{m}\over \partial x^{i}},{\partial f_{m}\over \partial y^{i}},{\partial f_{m}\over \partial z^{i}}\}\subset C_{c}(\mathcal{R}^{3})$, for $1\leq i\leq 13$\\

where $C_{c}(\mathcal{R}^{3})$ is the space of continuous functions with compact support, we have, for $m\in\mathcal{N}$;\\

$\mathcal{F}({\partial^{14}f_{m}\over \partial x^{14}}+{\partial^{14}g\over \partial y^{14}}+{\partial^{14}g\over \partial z^{14}})(\overline{k})$\\

$={1\over (2\pi)^{3\over 2}}\int_{-\infty}^{\infty}\int_{-\infty}^{\infty}\int_{-\infty}^{\infty}({\partial^{14}f_{m}\over \partial x^{14}}+{\partial^{14}f_{m}\over \partial y^{14}}+{\partial^{14}f_{m}\over \partial z^{14}})e^{-ik_{1}x}e^{-ik_{2}y}e^{-ik_{3}z}dxdydz$\\

$=((ik_{1})^{14}+(ik_{2})^{14}+(ik_{3})^{14}){1\over (2\pi)^{3\over 2}}\int_{-\infty}^{\infty}\int_{-\infty}^{\infty}\int_{-\infty}^{\infty}f_{m}(x,y,z)e^{-ik_{1}x}e^{-ik_{2}y}e^{-ik_{3}z}dxdydz$\\

$=(-k_{1}^{14}-k_{2}^{14}-k_{3}^{14})\mathcal{F}(f_{m})(\overline{k})$\\

so that, for $\overline{k}\neq \overline{0}$;\\

$|\mathcal{F}(f_{m})(\overline{k})|\leq {|\mathcal{F}({\partial^{14}f_{m}\over \partial x^{14}}+{\partial^{14}f_{m}\over \partial y^{14}}+{\partial^{14}f_{m}\over \partial z^{14}})(\overline{k})|\over (k_{1}^{14}+k_{2}^{14}+k_{3}^{14})}$ $(\dag)$\\

We have, using the result of Lemma \ref{existence2}, for sufficiently large $m$, that;\\

$|\mathcal{F}({\partial^{14}f_{m}\over \partial x^{14}}+{\partial^{14}f_{m}\over \partial y^{14}}+{\partial^{14}f_{m}\over \partial z^{14}})(\overline{k})|$\\

${1\over (2\pi)^{3\over 2}}|\int_{\mathcal{R}^{3}}({\partial^{14}f_{m}\over \partial x^{14}}+{\partial^{14}f_{m}\over \partial y^{14}}+{\partial^{14}f_{m}\over \partial z^{14}})e^{-ik_{1}x}e^{-ik_{2}y}e^{-ik_{3}z}dxdydz|$\\

$\leq{1\over (2\pi)^{3\over 2}}\int_{\mathcal{R}^{3}}(|{\partial f_{m}\over \partial x^{14}}|+|{\partial f_{m}\over \partial y^{14}}|+|{\partial f_{m}\over \partial z^{14}}|)dxdydz$\\

$\leq {3G\over (2\pi)^{3\over 2}}m^{3}$ $(\dag\dag)$\\

so that, combining $(\dag)$ and $(\dag\dag)$, we have, for $\overline{k}\neq \overline{0}$, sufficiently large $m$;\\

$|\mathcal{F}(f_{m})(\overline{k})|\leq {3G\over (2\pi)^{3\over 2}}{m^{3}\over (k_{1}^{14}+k_{2}^{14}+k_{3}^{14})}$ $(*)$\\

Using polar coordinates $k_{1}=rsin(\theta)cos(\phi)$, $k_{2}=rsin(\theta)sin(\phi)$, $k_{3}=rcos(\theta)$, $0\leq \theta\leq\pi$, $-\pi<\phi\leq \pi$, we have that;\\

${1\over (k_{1}^{14}+k_{2}^{14}+k_{3}^{14})}={1\over r^{14}}{1\over \alpha(\theta,\phi)}$\\

where $\alpha(\theta,\phi)=sin^{14}(\theta)(cos^{14}(\phi)+sin^{14}(\phi))+cos^{14}(\theta)$\\

We have that, in the range $0\leq \theta\leq\pi$, $-\pi\leq\phi\leq \pi$, with $\theta\neq {\pi\over 2}$, $|\phi|\neq {\pi\over 2}$;\\

$\alpha(\theta,\phi)=0$\\

iff $tan^{14}(\theta)(1+tan^{14}(\phi))+{1\over cos^{14}(\phi)}=0$\\

iff $tan^{14}(\theta)(1+tan^{14}(\phi))=-{1\over cos^{14}(\phi)}$\\

which has no solution, as the two sides of the equation have opposite signs.\\

and, with $\theta={\pi\over 2}$, , $|\phi|\neq {\pi\over 2}$\\

$\alpha(\theta,\phi)=0$\\

iff $cos^{14}(\phi)+sin^{14}(\phi)=0$\\

iff $tan^{14}(\phi)=-1$\\

which has no solution, as the two sides of the equation have opposite signs.\\

and, with $\theta\neq {\pi\over 2}$, , $|\phi|={\pi\over 2}$\\

$\alpha(\theta,\phi)=0$\\

iff $cos^{14}(\theta)+sin^{14}(\theta)=0$\\

iff $tan^{14}(\theta)=-1$\\

which has no solution, as the two sides of the equation have opposite signs.\\

and, with $\theta={\pi\over 2}$, , $|\phi|={\pi\over 2}$\\

$\alpha(\theta,\phi)=0$\\

iff $1=0$\\

which is not the case. It follows that $\alpha(\theta,\phi)=0$ has no solution in the range $0\leq \theta\leq\pi$, $-\pi\leq\phi\leq \pi$. By continuity, compactness of $[0\,\pi]\times [-\pi,\pi]$ and the fact that $\alpha({\pi\over 2},{\pi\over 2})=1$, restricting the interval $[-\pi,\pi]$, there exists $\epsilon>0$, with $\alpha(\theta,\phi)\geq \epsilon$, for $0\leq \theta\leq\pi$, $-\pi<\phi\leq \pi$. In particularly;\\

${1\over (k_{1}^{14}+k_{2}^{14}+k_{3}^{14})}\leq {1\over \epsilon r^{14}}$\\

$={1\over \epsilon |\overline{k}|^{14}}$\\

so that, from $(*)$;\\

$|\mathcal{F}(f_{m})(\overline{k})|\leq {3G\over (2\pi)^{3\over 2}}{m^{3}\over \epsilon |\overline{k}|^{14}}$\\

$={Dm^{3}\over |\overline{k}|^{14}}$\\

where $D={3G\over \epsilon (2\pi)^{3\over 2}}$\\

For the final claim, we have, for $1\leq i\leq 3$, $m\in\mathcal{N}$, as $f_{m}$ is supported on $W_{m+{1\over m}}$ and continuous, that $x_{i}f_{m}\in L^{1}(\mathcal{R}^{3})$ and, differentiating under the integral sign;\\

$|{\partial\mathcal{F}(f_{m})(\overline{k})\over \partial k^{i}}|=|{\partial \over \partial k^{i}}({1\over (2\pi)^{3\over 2}}\int_{\mathcal{R}^{3}}f_{m}(\overline{x})e^{-i\overline{k}\centerdot\overline{x}}d\overline{x})|$\\

$=|{-i\over (2\pi)^{3\over 2}}\int_{\mathcal{R}^{3}}x_{i}f_{m}(\overline{x})e^{-i\overline{k}\centerdot\overline{x}}d\overline{x})|$\\

$\leq {1\over (2\pi)^{3\over 2}}\int_{\mathcal{R}^{3}}|x_{i}f_{m}(\overline{x})|d\overline{x}$\\

$={1\over (2\pi)^{3\over 2}}||x_{i}f_{m}(\overline{x})||_{1}$\\

so that ${\partial\mathcal{F}(f_{m})(\overline{k})\over \partial k^{i}}$ is bounded, and, in particularly, $\mathcal{F}(f_{m})$ is continuous, for $m\in\mathcal{N}$. It follows, using the first result, and polar coordinates, that, for $n>1$, sufficiently large $m$;\\

$|\int_{\mathcal{R}^{3}}\mathcal{F}(f_{m})(\overline{k})d\overline{k}|\leq \int_{B(\overline{0},n)}|\mathcal{F}(f_{m})(\overline{k})|d\overline{k}+\int_{{\mathcal{R}^{3}\setminus B(\overline{0},n)}}|\mathcal{F}(f_{m})(\overline{k})|d\overline{k}$\\

$\leq {4C_{n}\pi^{3}\over 3}+\int_{{\mathcal{R}^{3}\setminus B(\overline{0},n)}}{Dm^{3}\over |\overline{k}|^{14}}$\\

$\leq {4C_{n}\pi^{3}\over 3}+\int_{0}^{\pi}\int_{-\pi}^{\pi}\int_{n}^{\infty}{Dm^{3}\over r^{14}}|r^{2}sin(\theta)|drd\theta d\phi$\\

$\leq {4C_{n}\pi^{3}\over 3}+2D\pi^{2}m^{3}\int_{n}^{\infty}{dr\over r^{12}}$\\

$\leq {4C_{n}\pi^{3}\over 3}+2D\pi^{2}m^{3}[{-1\over 11 r^{11}}]^{\infty}_{n}$\\

$={4C_{n}\pi^{3}\over 3}+{2D\pi^{2}m^{3}\over 11 n^{11}}$\\

where $C_{n}=||\mathcal{F}(f_{m})|_{B(\overline{0},n)}||_{\infty}$, so that $\mathcal{F}(f_{m})\in L^{1}(\mathcal{R}^{3})$.\\

A similar proof works in the two dimensional case.\\

\end{proof}

\begin{lemma}
\label{squarecube23}
Let $\{f_{m}:m\in\mathcal{N}\}$ be the inflexionary sequences constructed in Lemmas \ref{existence} and \ref{existence2}, then;\\

$\int_{[-m-{1\over m^{2}},m+{1\over m^{2}}]^{2}\setminus [-m,m]^{2}}|f_{m}|dxdy\leq {E\over m}$\\

for sufficiently large $m\in\mathcal{N}$, where $E\in\mathcal{R}_{>0}$.\\

$\int_{[-m-{1\over m^{3}},m+{1\over m^{3}}]^{3}\setminus [-m,m]^{3}}|f_{m}|dxdydz\leq {E\over m}$\\

for sufficiently large $m\in\mathcal{N}$, where $E\in\mathcal{R}_{>0}$.\\
\end{lemma}
\begin{proof}
By the construction, we obtain the result that for an inflexionary approximation sequence $f_{m}$ in $\mathcal{R}^{2}$ or $\mathcal{R}^{3}$;\\

$|f_{m}|_{[-m-{1\over m^{2}},m+{1\over m^{2}}]^{2}\setminus [-m,m]^{2}}\leq D$\\

$|f_{m}|_{[-m-{1\over m^{3}},m+{1\over m^{3}}]^{3}\setminus [-m,m]^{3}}\leq D$ $(*)$\\

independently of $m$. We give the proof of $(*)$ in the $3$-dimensional case. We have that, for $m\leq x\leq m+{1\over m^{3}}$, $m\leq y\leq m+{1\over m^{3}}$, $m\leq z\leq m+{1\over m^{3}}$;\\

$|f_{m}|(x,y,z)\leq \sum_{i=0}^{13}D_{i}|{\partial^{i}f_{m}\over \partial z^{i}}|(x,y,m)$\\

$\leq \sum_{i=0}^{13}D_{i}\sum_{j=0}^{13}D_{ij}{\partial^{i+j}f_{m}\over \partial y^{j}\partial z^{i}}|(x,m,m)$\\

$\leq \sum_{i=0}^{13}D_{i}\sum_{j=0}^{13}D_{ij}\sum_{k=0}^{13}D_{ijk}{\partial^{i+j+k}f_{m}\over \partial x^{k}\partial y^{j}\partial z^{i}}|(m,m,m)$\\

$=\sum_{i=0}^{13}D_{i}\sum_{j=0}^{13}D_{ij}\sum_{k=0}^{13}D_{ijk}{\partial^{i+j+k}f\over \partial x^{k}\partial y^{j}\partial z^{i}}|(m,m,m)$\\

$\leq C\sum_{i,j,k=0}^{13}D_{i}D_{ij}D_{ijk}$\\

$=C\sum_{i,j,k=0}^{13}D_{i}D_{j}D_{k}=D$\\

The proof of the bound for the other regions is similar and left to the reader, as is the two dimensional case. It follows that, using the binomial theorem;\\

$\int_{[-m-{1\over m^{2}},m+{1\over m^{2}}]^{2}\setminus [-m,m]^{2}}|f_{m}|dxdy$\\

 $\leq Darea([-m-{1\over m^{2}},m+{1\over m^{2}}]^{2}\setminus [-m,m]^{2})$\\

 $=4D((m+{1\over m^{2}})^{2}-m^{2})$\\

 $4D(m^{2}+{2m\over m^{2}}+{1\over m^{4}}-m^{2})$\\

 $\leq {E\over m}$\\

 and;\\

 $\int_{[-m-{1\over m^{3}},m+{1\over m^{3}}]^{3}\setminus [-m,m]^{3}}|f_{m}|dxdydz$\\

 $\leq Dvol([-m-{1\over m^{3}},m+{1\over m^{3}}]^{3}\setminus [-m,m]^{3})$\\

 $=8D((m+{1\over m^{3}})^{3}-m^{3})$\\

 $8D(m^{3}+{3m^{2}\over m^{3}}+{3m\over m^{6}}+{1\over m^{9}}-m^{3})$\\

 $\leq {E\over m}$\\

for $m$ sufficiently large, where $E\in\mathcal{R}_{>0}$.\\

\end{proof}

\begin{lemma}
\label{sequences23}
Let $f\in C^{\infty}(\mathcal{R}^{3})$ be quasi split normal, with the Fourier transform $\mathcal{F}$ defined in \cite{dep3}. Let $\{f_{m}:m\in\mathcal{N}\}$ be the inflexionary sequence constructed in Lemma \ref{existence2}. Let $\mathcal{F}$ be the ordinary Fourier transform, defined for each $f_{m}$, then, for any $(k_{01},k_{02},k_{03})$, with $k_{01}\neq 0$, $k_{02}\neq 0$, $k_{03}\neq 0$, the sequence $\{\mathcal{F}(f_{m}):m\in\mathcal{N}\}$ converges pointwise and uniformly to $\mathcal{F}(f)$ on ${\mathcal{R}^{3}\setminus (|k_{1}|<k_{01})\cup (|k_{2}|<k_{02})\cup (|k_{3}|<k_{03})}$. In particularly, $\mathcal{F}(f)\in C({\mathcal{R}^{3}\setminus\{k_{1}=0\cup k_{2}=0\cup k_{3}=0\}})$. A corresponding result holds in dimension $2$.\\

\end{lemma}
\begin{proof}
For $g\in C_{c}(\mathcal{R}^{3})$ or $g$ quasi split normal, and $m\in\mathcal{N}$, define;\\

$\mathcal{F}_{m}(g)(\overline{k})={1\over (2\pi)^{3\over 2}}\int_{C_{m}}g(\overline{x})e^{-i\overline{k}\centerdot\overline{x}}d\overline{x}$\\

For $\overline{k}\in{\mathcal{R}^{3}\setminus (|k_{1}|<k_{01})\cup (|k_{2}|<k_{02})\cup (|k_{3}|<k_{03})}$, $m\in \mathcal{N}$, $\epsilon>0$, we have, using Lemma \ref{squarecube23};\\

$|\mathcal{F}(f)(\overline{k})-\mathcal{F}(f_{m})(\overline{k})|\leq |\mathcal{F}(f)(\overline{k})-\mathcal{F}_{m}(f)(\overline{k})|+|\mathcal{F}_{m}(f)(\overline{k})-\mathcal{F}_{m}(f_{m})(\overline{k})|$\\

$+|\mathcal{F}_{m}(f_{m})(\overline{k})-\mathcal{F}(f_{m})(\overline{k})|$\\

$=|\mathcal{F}(f)(\overline{k})-\mathcal{F}_{m}(f)(\overline{k})|+|\mathcal{F}_{m}(f_{m})(\overline{k})-\mathcal{F}(f_{m})(\overline{k})|$\\

$\leq |\mathcal{F}(f)(\overline{k})-\mathcal{F}_{m}(f)(\overline{k})|+|\int_{\mathcal{R}^{3}\setminus C_{m}}f_{m}(\overline{x})e^{-i\overline{k}\centerdot\overline{x}}d\overline{x}|$\\

$\leq |\mathcal{F}(f)(\overline{k})-\mathcal{F}_{m}(f)(\overline{k})|+\int_{C_{m+{1\over m^{3}}}\setminus C_{m}}|f_{m}(\overline{x})|d\overline{x}$\\

$\leq |\mathcal{F}(f)(\overline{k})-\mathcal{F}_{m}(f)(\overline{k})|+{E\over m}$ $(BB)$\\

By the result in \cite{dep3}, we have that, for sufficiently large $m$;\\

$|\mathcal{F}(f)(\overline{k})-\mathcal{F}_{m}(f)(\overline{k})|\leq {C_{k_{01},k_{02},k_{03}}\over m}$ $(B)$\\

Combining $(B)$ and $(BB)$, we obtain that;\\

$|\mathcal{F}(f)(\overline{k})-\mathcal{F}(f_{m})(\overline{k})|\leq {C_{k_{01},k_{02},k_{03}}+E\over m}$\\

$\leq \epsilon$\\

for $m\geq {C_{k_{01},k_{02},k_{03}}+E\over\epsilon}$. As $\epsilon>0$ was arbitrary, we obtain the first result. The fact that each $\mathcal{F}(f_{m})$ is continuous, follows from the differentiability $\mathcal{F}(f_{m})$, which is a consequence of the fact that $x_{i}f_{m}(\overline{x}$ has compact support, for $1\leq i\leq 3$. The last result then follows immediately from the fact that $k_{01}\neq 0$, $k_{02}\neq 0$, $k_{03}\neq 0$ were arbitrary and the uniform limit of continuous functions is continuous. The last claim is similar.\\

\end{proof}

\begin{lemma}
\label{inversion3D}
Let $f\in C^{\infty}(\mathcal{R}^{3})$, with ${\partial^{i_{1}+i_{2}+i_{3}}\over \partial x^{i_{1}}\partial y^{i_{2}}\partial z^{i_{3}}}$ bounded for $0\leq i_{1}+i_{2}+i_{3}\leq 40$, $f$ quasi split normal, and of moderate decrease. Then;\\

$f(\overline{x})=\mathcal{F}^{-1}(\mathcal{F}(f))(\overline{x})$, $(\overline{x}\in\mathcal{R}^{3})$\\

where, for $g\in L^{1}(\mathcal{R}^{3})$;\\

$\mathcal{F}^{-1}(g)(\overline{x})={1\over (2\pi)^{3\over 2}}\int_{\mathcal{R}^{3}}g(\overline{k})e^{i\overline{k}\centerdot\overline{x}}d\overline{k}$\\

The same claim holds in dimension $2$.\\
\end{lemma}

\begin{proof}
By Lemma \ref{integrablefirst}, we have that $\mathcal{F}(f)\in L^{1}(\mathcal{R}^{3})$. Let $\{f_{m}:m\in\mathcal{N}\}$ be the inflexionary approximating sequence, given by Lemma \ref{existence}, then, for sufficiently large m, $f_{m}\in L^{1}(\mathcal{R}^{3})$ and $\mathcal{F}(f_{m})\in L^{1}(\mathcal{R}^{3})$ by Lemma \ref{decayofapprox}. It follows, see \cite{fol} or the method of \cite{dep2}, that for such $m$, $f_{m}=\mathcal{F}^{-1}(\mathcal{F}(f_{m}))$, $(***)$, By the proof of Lemma \ref{sequences23}, we have that, for $\overline{k}$ with $min(|k_{1}|,|k_{2}|,|k_{3}|)>\epsilon>0$, $|\mathcal{F}(f)(k)-\mathcal{F}(f_{m})(k)|\leq {E_{\epsilon}\over m}$, (B). By the fact that $f$ is of very moderate decrease, we have that $\mathcal{F}(f)-\mathcal{F}(f_{m})\in L^{2}(\mathcal{R}^{3})$, and by the classical theory, and by the proof of Lemma \ref{squarecube23}, we have that;\\

$||\mathcal{F}(f)-\mathcal{F}(f_{m})||^{2}_{L^{2}(\mathcal{R}^{3})}$\\

$=||f-f_{m}||^{2}_{L^{2}(\mathcal{R}^{3})}$\\

$\leq \int_{\mathcal{R}^{3}\setminus C_{m}}|f|^{2}d\overline{x}+\int_{C_{m+{1\over m^{3}}}\setminus C_{m}}|f_{m}|^{2}d\overline{x}$\\

$\leq \int_{\mathcal{R}^{3}\setminus B(\overline{0},m)}|f|^{2}d\overline{x}+{G\over m}$\\

$\leq \int_{\mathcal{R}^{3}\setminus B(\overline{0},m)}{C\over |\overline{x}|^{4}}d\overline{x}+{G\over m}$\\

$\leq 2\pi^{2}\int_{m}^{\infty}{C\over r^{2}}dr+{G\over m}$\\

$\leq {C\over m}+{G\over m}$\\

$\leq {F\over m}$\\

where $\{C,F,G\}\subset\mathcal{R}_{>0}$. It follows that $||\mathcal{F}(f)-\mathcal{F}(f_{m})||_{L^{2}(\mathcal{R}^{3})}\rightarrow 0$ as $m\rightarrow \infty$. In particularly, there exists a constant $H\in\mathcal{R}_{>0}$ with $||\mathcal{F}(f)-\mathcal{F}(f_{m})||_{L^{2}(\mathcal{R}^{3})}\leq H$, for sufficiently large $m$. By the Cauchy Schwarz inequality, we have that, for $m$ sufficiently large;\\

$||\mathcal{F}(f)-\mathcal{F}(f_{m})||_{L^{1}(B(\overline{0},n))}$\\

$\leq ||(\mathcal{F}(f)-\mathcal{F}(f_{m}))|_{B(\overline{0},n)}||_{L^{2}(B(\overline{0},n))}||1_{B(\overline{0},n)}||_{L^{2}(B(\overline{0},n))}$\\

$\leq {\sqrt{F}\over \sqrt{m}}||1_{B(\overline{0},n)}||_{L^{2}(B(\overline{0},n))}$\\

$={2\sqrt{F\pi}n^{3\over 2}\over \sqrt{3m}}$

$={Kn^{3\over 2}\over m^{1\over 2}}$, $(A)$\\

Using the fact from Lemma \ref{integrablefirst}, that $\mathcal{F}(f)\in L^{1}(\mathcal{R})$, and of rapid decrease, for $\delta>0$ arbitrary, we have that;\\

$\int_{{\mathcal{R}^{3}\setminus B(\overline{0},n)}}|\mathcal{F}(f)(\overline{k}|d\overline{k}<\delta$\\

for $n\in\mathcal{N}$, sufficiently large, $n\geq n_{0}$. Choosing $n\in\mathcal{N}$, with $m=[n^{10\over 3}]$, and using $(A)$, Lemma \ref{decayofapprox}, we have, for $\overline{x}\in\mathcal{R}^{3}$, that;\\

$|\mathcal{F}^{-1}(\mathcal{F}(f))(\overline{x})-\mathcal{F}^{-1}(\mathcal{F}(f_{m}))(\overline{x})|=|\mathcal{F}^{-1}(\mathcal{F}(f)(\overline{k})-\mathcal{F}(f_{m})(\overline{k}))|$\\

$={1\over (2\pi)^{3\over 2}}|\int_{B(\overline{0},n)}(\mathcal{F}(f)(\overline{k})-\mathcal{F}(f_{m})(\overline{k}))e^{i\overline{k}\centerdot\overline{x}}d\overline{k}$\\

$+\int_{{\mathcal{R}^{3}\setminus B(\overline{0},n)}}(\mathcal{F}(f)(\overline{k})-\mathcal{F}(f_{m})(\overline{k}))e^{i\overline{k}\centerdot\overline{x}}d\overline{k}|$\\

$\leq {1\over (2\pi)^{3\over 2}}(\int_{B(\overline{0},n)}|\mathcal{F}(f)(\overline{k})-\mathcal{F}(f_{m})(\overline{k})|d\overline{k}$\\

$+\int_{{\mathcal{R}^{3}\setminus B(\overline{0},n)}}|\mathcal{F}(f)(\overline{k})|d\overline{k}+\int_{{\mathcal{R}^{3}\setminus B(\overline{0},n)}}|\mathcal{F}(f_{m})(\overline{k})|d\overline{k})$\\

$\leq {1\over (2\pi)^{3\over 2}}(\int_{B(\overline{0},n)}|\mathcal{F}(f)(\overline{k})-\mathcal{F}(f_{m})(\overline{k})|d\overline{k}+\delta+\int_{{\mathcal{R}^{3}\setminus B(\overline{0},n)}}{Dm^{3}\over |k|^{14}}d\overline{k})$\\

$\leq {1\over (2\pi)^{3\over 2}}({Kn^{3\over 2}\over 3m^{1\over 2}}+\delta+\int_{{\mathcal{R}^{3}\setminus B(\overline{0},n)}}{Dm^{3}\over |k|^{14}}d\overline{k})$\\

$\leq {1\over (2\pi)^{3\over 2}}({Kn^{3\over 2}\over [n^{10\over 3}]^{1\over 2}}+\delta+\int_{{\mathcal{R}^{3}\setminus B(\overline{0},n)}}{Dn^{10}\over |k|^{14}}d\overline{k})$\\

$\leq {1\over (2\pi)^{3\over 2}}({K\over n^{1\over 6}}+\delta+2\pi^{2}\int_{r>n}{Dn^{10}\over r^{14}}dr)$\\

$= {1\over (2\pi)^{3\over 2}}({K\over n^{1\over 6}}+\delta+2D\pi^{2}n^{10}[{-1\over 13r^{13}}]^{\infty}_{n})$\\

$={1\over (2\pi)^{3\over 2}}({K\over n^{1\over 6}}+\delta+{2D\pi^{2}\over 13n^{3}})$\\

$<{2\delta\over (2\pi)^{3\over 2}}$\\

for sufficiently large $n\geq n_{0}$, or $m\geq m_{0}$, so that, as $\epsilon>0$ and $\delta>0$ were arbitrary, for $\overline{x}\in\mathcal{R}^{3}$;\\

$lim_{m\rightarrow\infty}\mathcal{F}^{-1}(\mathcal{F}(f_{m}))(\overline{x})=\mathcal{F}^{-1}\mathcal{F}(f)(\overline{x})$, $(****)$\\

and, by Definition \ref{threedimensions}, $(***)$, $(****)$;\\

$f(\overline{x})=lim_{m\rightarrow\infty}f_{m}(\overline{x})=lim_{m\rightarrow\infty}\mathcal{F}^{-1}(\mathcal{F}(f_{m}))(\overline{x})=\mathcal{F}^{-1}\mathcal{F}(f)(\overline{x})$\\

The proof of the final claim in dimension $2$ is identical.\\

\end{proof}
The following results are not required for the proof of the inversion theorem but are required in \cite{dep2}.\\

\begin{defn}
\label{moderate3D}
We say that $f:\mathcal{R}^{3}\rightarrow\mathcal{R}$ is of very moderate decrease if $|f(\overline{x})|\leq {C\over |\overline{x}|}$ for $|\overline{x}|>C$, $C\in\mathcal{R}_{>0}$. We say that $f:\mathcal{R}^{3}\rightarrow\mathcal{R}$ is of moderate decrease $n$ if $|f(\overline{x})|\leq {C\over |\overline{x}|^{n}}$ for $|\overline{x}|>C$, $C\in\mathcal{R}_{>0}$, $n\geq 2$. We just say that $f$ is of moderate decrease if $f$ is of moderate decrease $2$. We call $\{\theta,\phi\}$ generic if $sin(\theta)cos(\phi)\neq 0$, $sin(\theta)sin(\phi)\neq 0$, $cos(\theta)\neq 0$\\
\end{defn}

\begin{lemma}
\label{differentiableinversion}
Let $f$ be of very moderate decrease and quasi split normal, $f\in C^{41}(\mathcal{R}^{3})$, such that the partial derivatives $\{{\partial f^{i+j+k}\over \partial x^{i}\partial y^{j}\partial z^{k}}:1\leq i+j+k\leq 41\}$ are of moderate decrease, and of moderate decrease $i+j+k+1$, then for $1\leq i\leq 3$;\\

$k_{i}\mathcal{F}(f)(\overline{k})\in C^{1}({\mathcal{R}^{3}\setminus (k_{1}=0\cup k_{2}=0 \cup k_{3}=0)})$\\

$lim_{\overline{k}\rightarrow 0,\overline{k}\notin (k_{1}=0\cup k_{2}=0 \cup k_{3}=0)}k_{i}\mathcal{F}(f)(\overline{k})=0$\\

The same results hold for $k_{i}\mathcal{F}({\partial f\over \partial x_{j}})$, $1\leq i\leq j\leq 3$, when $f\in C^{42}(\mathcal{R}^{3})$.\\

Making a polar coordinate change, for $\{\theta,\phi\}$ generic,  $r\mathcal{F}(f)_{\theta,\phi}(r)\in C^{1}(\mathcal{R}_{>0})$, $lim_{r\rightarrow 0}r\mathcal{F}(f)_{\theta,\phi}(r)=0$, and similarly for $r\mathcal{F}({\partial f\over \partial x_{j}})$, $1\leq j\leq 3$.\\

We have that $\mathcal{F}(f)(\overline{k})\in L^{1}(\mathcal{R}^{3})$, $\{{\mathcal{F}({\partial f\over \partial x_{j}})(\overline{k})\over |\overline{k}|}:1\leq j\leq 3\}\subset L^{1}(\mathcal{R}^{3})$\\

and $\{{\mathcal{F}({\partial^{2} f\over \partial x_{i}\partial x_{j}})(\overline{k})\over |\overline{k}|^{2}}:1\leq i,j\leq 3\}\subset L^{1}(\mathcal{R}^{3})$\\

For any given $\epsilon>0$, there exists $\delta>0$, for $1\leq j\leq 3$, such that for a generic translation $\overline{l}$ with $l_{1}\neq 0$, $l_{2}\neq 0$, $l_{3}\neq 0$;\\

$max(|\int_{0}^{\delta}r\mathcal{F}_{\theta,\phi,\overline{l}}({\partial f\over \partial x_{j}})(r)dr|,|\int_{0}^{\delta}{d\over dr}(r\mathcal{F}_{\theta,\phi,\overline{l}}({\partial f\over \partial x_{j}})(r))dr|)<\epsilon$\\

uniformly in $\{\theta,\phi\}$.\\

\end{lemma}
\begin{proof}
As ${\partial f\over \partial x}$ is of moderate decrease and quasi split normal, for fixed $y,z$, $f_{y,z}$ is of very moderate decrease and analytic at infinity, we have for $k_{1}\neq 0$, $k_{2}\neq 0$, $k_{3}\neq 0$;\\

$\mathcal{F}({\partial f\over \partial x})={1\over (2\pi)^{3\over 2}}lim_{r_{1}\rightarrow\infty}lim_{r_{2}\rightarrow \infty}lim_{r_{3}\rightarrow\infty}\int_{-r_{1}}^{r_{1}}\int_{-r_{2}}^{r_{2}}\int_{-r_{3}}^{r_{3}}{\partial f\over \partial x}(\overline{x})e^{-i\overline{k}\centerdot \overline{x}}dx_{1}dx_{2}dx_{3}$\\

$={1\over (2\pi)^{3\over 2}}lim_{r_{2}\rightarrow\infty}lim_{r_{3}\rightarrow \infty}\int_{-r_{2}}^{r_{2}}\int_{-r_{3}}^{r_{3}}(lim_{r_{1}\rightarrow\infty}\int_{-r_{1}}^{r_{1}}{\partial f\over \partial x}(\overline{x})e^{-ik_{1}x_{1}}dx_{1})e^{-i(k_{2}x_{2}+k_{3}x_{3})}dx_{2}dx_{3}$\\

$={1\over (2\pi)^{3\over 2}}lim_{r_{2}\rightarrow\infty}lim_{r_{3}\rightarrow \infty}\int_{-r_{2}}^{r_{2}}\int_{-r_{3}}^{r_{3}}(lim_{r_{1}\rightarrow\infty}([fe^{-ikx_{1}}]_{-r_{1}}^{r_{1}}+ik_{1}\int_{-r_{1}}^{r_{1}}f(\overline{x})e^{-ikx_{1}}dx_{1})$\\

$e^{-i(k_{2}x_{2}+k_{3}x_{3})}dx_{2}dx_{3}$\\

$=ik_{1}{1\over (2\pi)^{3\over 2}}lim_{r_{2}\rightarrow\infty}lim_{r_{3}\rightarrow \infty}\int_{-r_{2}}^{r_{2}}\int_{-r_{3}}^{r_{3}}(lim_{r_{1}\rightarrow\infty}\int_{-r_{1}}^{r_{1}}f(\overline{x})e^{-ikx_{1}}dx_{1})e^{-i(k_{2}x_{2}+k_{3}x_{3})}dx_{2}dx_{3}$\\

$=ik_{1}{1\over (2\pi)^{3\over 2}}lim_{r_{1}\rightarrow\infty}lim_{r_{2}\rightarrow \infty}lim_{r_{3}\rightarrow\infty}\int_{-r_{1}}^{r_{1}}\int_{-r_{2}}^{r_{2}}\int_{-r_{3}}^{r_{3}}f(\overline{x})e^{-i\overline{k}\centerdot \overline{x}}dx_{1}dx_{2}dx_{3}$\\

$=ik_{1}\mathcal{F}(f)(\overline{k})$ $(TT)$\\

the limit interchange being justified by the calculation in \cite{dep3}. It follows that, for $k_{1}\neq 0$, $k_{2}\neq 0$, $k_{3}\neq 0$, we have that;\\

$k_{1}\mathcal{F}(f)(\overline{k})=-i\mathcal{F}({\partial f\over \partial x})$\\

and similarly;\\

$k_{i}\mathcal{F}(f)(\overline{k})=-i\mathcal{F}({\partial f\over \partial x_{i}})$ $(A)$, for $1\leq i\leq 3$ and $k_{1}\neq 0$, $k_{2}\neq 0$, $k_{3}\neq 0$.\\

It follows that, using the fact that;\\

$F(x_{1},k_{2},k_{3})=lim_{r_{2}\rightarrow\infty}lim_{r_{3}\rightarrow\infty}\int_{-r_{2}}^{r_{2}}\int_{-r_{3}}^{r_{3}}{\partial f\over \partial x}(x_{1},x_{2},x_{3})e^{-ik_{2}x_{2}}e^{-ik_{3}x_{3}}dx_{2}dx_{3}$\\

is of moderate decrease, the DCT and the FTC, and the fact that $f_{y,z}$ is of very moderate decrease;\\

$lim_{\overline{k}\rightarrow 0,\overline{k}\notin (k_{1}=0\cup k_{2}=0 \cup k_{3}=0)}k_{1}\mathcal{F}(f)(\overline{k})$\\

$-ilim_{\overline{k}\rightarrow 0,\overline{k}\notin (k_{1}=0\cup k_{2}=0 \cup k_{3}=0)}\mathcal{F}(f)({\partial f\over \partial x})(\overline{k})$\\

$={-i\over (2\pi)^{3\over 2}}lim_{\overline{k}\rightarrow 0,\overline{k}\notin (k_{1}=0\cup k_{2}=0 \cup k_{3}=0)}lim_{r_{1}\rightarrow\infty}lim_{r_{2}\rightarrow \infty}lim_{r_{3}\rightarrow\infty}\int_{-r_{1}}^{r_{1}}\int_{-r_{2}}^{r_{2}}\int_{-r_{3}}^{r_{3}}{\partial f\over \partial x}(\overline{x})e^{-i\overline{k}\centerdot \overline{x}}dx_{1}dx_{2}dx_{3}$\\

$={1\over (2\pi)^{3\over 2}}lim_{k_{2}\rightarrow 0,k_{3}\rightarrow 0,k_{2}\neq 0,k_{3}\neq 0}lim_{r_{2}\rightarrow\infty}lim_{r_{3}\rightarrow \infty}\int_{-r_{2}}^{r_{2}}\int_{-r_{3}}^{r_{3}}(lim_{k_{1}\rightarrow 0}\int_{-\infty}^{\infty}{\partial f\over \partial x}(\overline{x})e^{-ik_{1}x_{1}}dx_{1})$\\

$e^{-i(k_{2}x_{2}+k_{3}x_{3})}dx_{2}dx_{3}$\\

$={1\over (2\pi)^{3\over 2}}lim_{k_{2}\rightarrow 0,k_{3}\rightarrow 0,k_{2}\neq 0,k_{3}\neq 0}lim_{r_{2}\rightarrow\infty}lim_{r_{3}\rightarrow \infty}\int_{-r_{2}}^{r_{2}}\int_{-r_{3}}^{r_{3}}(\int_{-\infty}^{\infty}{\partial f\over \partial x}(\overline{x})dx_{1})e^{-i(k_{2}x_{2}+k_{3}x_{3})}dx_{2}dx_{3}$\\

$={1\over (2\pi)^{3\over 2}}lim_{k_{2}\rightarrow 0,k_{3}\rightarrow 0,k_{2}\neq 0,k_{3}\neq 0}lim_{r_{2}\rightarrow\infty}lim_{r_{3}\rightarrow \infty}\int_{-r_{2}}^{r_{2}}\int_{-r_{3}}^{r_{3}}([f]_{-\infty}^{\infty})e^{-i(k_{2}x_{2}+k_{3}x_{3})}dx_{2}dx_{3}$\\

$=0$ $(E)$\\

Similarly;\\

$lim_{\overline{k}\rightarrow 0,\overline{k}\notin (k_{1}=0\cup k_{2}=0 \cup k_{3}=0)}k_{i}\mathcal{F}(f)(\overline{k})=0$, $1\leq i\leq 3$\\

As $f\in C^{41}(\mathcal{R}^{3})$, we have, by the product rule, that $x_{i}{\partial f\over \partial x_{j}}\in C^{40}(\mathcal{R}^{3})$, $1\leq i\leq j\leq 3$. As $f$ is of very moderate decrease and;\\

$\{{\partial f^{l+m+n}\over \partial x_{1}^{l}\partial x_{2}^{m}\partial x_{3}^{m}}:1\leq l+m+n\leq 40\}$\\

are of very moderate decrease, we have, by repeated application of the product rule again, that;\\

$\{{\partial^{l+m+n}x_{i}{\partial f\over \partial x_{j}}\over \partial x_{1}^{l}\partial x_{2}^{m}\partial x_{3}^{n}}:0\leq l+m+n\leq 40\}$, $1\leq i\leq j\leq 3$\\

are bounded. By Lemma \ref{squarecube}, there exists an inflexionary approximation sequence $g_{m}$ for $x{\partial f\over \partial x}$ with the properties that;\\

$(i)$ $g_{m}\in C^{13,13,14}(\mathcal{R}^{3})$\\

$(ii)$. $g_{m}|_{[-m,m]^{3}}=x{\partial f\over \partial x}|_{[-m,m]^{3}}$\\

$(iii)$. $\int_{{[-m-{1\over m^{3}},m+{1\over m^{3}}]^{3}\setminus [-m,m]^{3}}}|g_{m}(\overline{x})d\overline{x}\leq {E\over m}$\\

$(iv)$. $g_{m}|_{{\mathcal{R}^{3}\setminus [-m-{1\over m^{3}},m+{1\over m^{3}}]^{3}}}=0$\\

By the construction of $g_{m}$, we have that $f_{m}={g_{m}\over x}$ is an approximation sequence for ${\partial f\over \partial x}$, with the property that;\\

$(i)'$ $f_{m}\in C^{13,13,14}(\mathcal{R}^{3})$\\

$(ii)'$. $f_{m}|_{[-m,m]^{3}}={\partial f\over \partial x}|_{[-m,m]^{3}}$\\

$(iii)'$. $\int_{{[-m-{1\over m^{3}},m+{1\over m^{3}}]^{3}\setminus [-m,m]^{3}}}|f_{m}(\overline{x})d\overline{x}\leq {E'\over m}$\\

$(iv)'$. $f_{m}|_{{\mathcal{R}^{3}\setminus [-m-{1\over m^{3}},m+{1\over m^{3}}]^{3}}}=0$\\

Following through the proof of Lemma \ref{sequences23}, as ${\partial f\over \partial x}$ is quasi split normal of moderate decrease and, therefore, of very moderate decrease, we have that $\mathcal{F}(f_{m})$ converges uniformly to $\mathcal{F}({\partial f\over \partial x})$ on compact subsets of ${\mathcal{R}^{3}\setminus (k_{1}=0\cup k_{2}=0\cup k_{3}=0)}$, so that $\mathcal{F}({\partial f\over \partial x})\in C({\mathcal{R}^{3}\setminus (k_{1}=0\cup k_{2}=0\cup k_{3}=0)})$, As $x_{i}x_{j}f_{m}\in L^{1}(\mathcal{R}^{3})$, for $1\leq i\leq j\leq 3$, we have that $\mathcal{F}(f_{m})$ is twice differentiable, in particularly, $\mathcal{F}(f_{m})\in C^{1}(\mathcal{R}^{3})$. As $f$ is quasi split normal, so is ${\partial f\over \partial x}$ and $x{\partial f\over \partial x}$. It follows that for $\{m,n\}\subset\mathcal{N}$, with $m\geq n$, differentiating under the integral sign, using the DCT, property $(iii)$ of an inflexionary approximating sequence, and the fact that $x{\partial f\over \partial x}$ is of very moderate decrease and quasi split normal, for $|k_{1}|\geq \epsilon_{1}>0$, $|k_{2}|\geq \epsilon_{2}>0$, $|k_{3}|\geq \epsilon_{3}>0$, we have that;\\

$|{\partial \mathcal{F}(f_{m})\over \partial k_{1}}-{\partial \mathcal{F}(f_{n})\over \partial k_{1}}|$\\

$={1\over (2\pi)^{3\over 2}}|{\partial \over \partial k_{1}}(\int_{\mathcal{R}^{3}}f_{m}(\overline{x})e^{-i\overline{k}\centerdot \overline{x}}d\overline{x}-{\partial \over \partial k_{1}}\int_{\mathcal{R}^{3}}f_{n}(\overline{x})e^{-i\overline{k}\centerdot \overline{x}}d\overline{x}|$\\

$={1\over (2\pi)^{3\over 2}}|\int_{\mathcal{R}^{3}}-ix_{1}f_{m}(\overline{x})e^{-i\overline{k}\centerdot \overline{x}}d\overline{x}-\int_{\mathcal{R}^{3}}-ix_{1}f_{n}(\overline{x})e^{-i\overline{k}\centerdot \overline{x}}d\overline{x}|$\\

$={1\over (2\pi)^{3\over 2}}|\int_{\mathcal{R}^{3}}(g_{m}-g_{n})(\overline{x})e^{-i\overline{k}\centerdot \overline{x}}d\overline{x}|$\\

$\leq {1\over (2\pi)^{3\over 2}}(\int_{{[-m-{1\over m^{3}},m+{1\over m^{3}}]^{3}\setminus [-m,m]^{3}}}|g_{m}(\overline{x})|d\overline{x}+\int_{{[-m-{1\over m^{3}},m+{1\over m^{3}}]^{3}\setminus [-m,m]^{3}}}|g_{n}(\overline{x})|d\overline{x}$\\

$+|\int_{{[-m,m]^{3}\setminus [-n,n]^{3}}}x_{1}{\partial f\over \partial x_{1}}e^{-i\overline{k}\centerdot\overline{x}}d\overline{x}|)$\\

$\leq {E\over m}+{E\over n}+{C(\overline{k})\over n}$ $(*)$\\

where $C(\overline{k})$ is uniformly bounded on the region $|k_{1}|\geq \epsilon_{1}>0$, $|k_{2}|\geq \epsilon_{2}>0$, $|k_{3}|\geq \epsilon_{3}>0$. It follows that the sequence $\{{\partial \mathcal{F}(f_{m})\over \partial k_{1}}:m\in\mathcal{N}\}$ is uniformly Cauchy on the region $|k_{1}|\geq \epsilon_{1}>0$, $|k_{2}|\geq \epsilon_{2}>0$, $|k_{3}|\geq \epsilon_{3}>0$, and converges uniformly. By considering inflexionary sequences for $y{\partial f\over \partial x}$ and $z{\partial f\over \partial x}$, we can similarly show that the sequences $\{{\partial \mathcal{F}(f_{m})\over \partial k_{2}}:m\in\mathcal{N}\}$ and $\{{\partial \mathcal{F}(f_{m})\over \partial k_{3}}:m\in\mathcal{N}\}$ are uniformly Cauchy on the region $|k_{1}|\geq \epsilon_{1}>0$, $|k_{2}|\geq \epsilon_{2}>0$, $|k_{3}|\geq \epsilon_{3}>0$, and converge uniformly. As $\mathcal{F}(f_{m})$ converges uniformly to $\mathcal{F}({\partial f\over \partial x})$ on the regions $|k_{1}|\geq \epsilon_{1}>0$, $|k_{2}|\geq \epsilon_{2}>0$, $|k_{3}|\geq \epsilon_{3}>0$, it follows that $\mathcal{F}({\partial f\over \partial x})\in C^{1}({\mathcal{R}^{3}\setminus (k_{1}=0\cup k_{2}=0\cup k_{3}=0)})$. The same result folds for $\mathcal{F}({\partial f\over \partial y})$ and $\mathcal{F}({\partial f\over \partial z})$, so by $(A)$;\\

 $\{k_{1}\mathcal{F}(f)(\overline{k}),k_{2}\mathcal{F}(f)(\overline{k}),k_{3}\mathcal{F}(f)(\overline{k})\}\subset C^{1}({\mathcal{R}^{3}\setminus (k_{1}=0\cup k_{2}=0\cup k_{3}=0)})$\\

 $(B)$\\

It follows that, changing to polars;\\

${\partial r\mathcal{F}(f)(\overline{k})\over \partial r}=({\partial\over \partial k_{1}}{k_{1}\over r}+{\partial\over \partial k_{2}}{k_{2}\over r}+{\partial\over \partial k_{3}}{k_{3}\over r})(r\mathcal{F}(f)(\overline{k}))$\\

$={\partial k_{1}\mathcal{F}(f)(\overline{k})\over \partial k_{1}}+{\partial k_{2}\mathcal{F}(f)(\overline{k})\over \partial k_{2}}+{\partial k_{3}\mathcal{F}(f)(\overline{k})\over \partial k_{3}}$ $(WW)$\\

so that, for generic $\{\theta,\phi\}$, $r\mathcal{F}(f)(r)_{\theta,\phi}\in C^{1}(\mathcal{R}_{>0})$, by $(B)$. Moreover;\\

$lim_{r\rightarrow 0}r\mathcal{F}(f)(r)_{\theta,\phi}$.\\

$=lim_{\overline{k}(\theta,\phi)\rightarrow 0}{r\over k_{1}}lim_{\overline{k}(\theta,\phi)\rightarrow \overline{0},k_{1}\neq 0,k_{2}\neq 0,k_{3}\neq 0}k_{1}\mathcal{F}(f)(\overline{k})$\\

$=lim_{\overline{k}(\theta,\phi)\rightarrow 0}{r\over k_{2}}lim_{\overline{k}(\theta,\phi)\rightarrow \overline{0},k_{1}\neq 0,k_{2}\neq 0,k_{3}\neq 0}k_{2}\mathcal{F}(f)(\overline{k})$\\

$=lim_{\overline{k}(\theta,\phi)\rightarrow 0}{r\over k_{3}}lim_{\overline{k}(\theta,\phi)\rightarrow \overline{0},k_{1}\neq 0,k_{2}\neq 0,k_{3}\neq 0}k_{3}\mathcal{F}(f)(\overline{k})$\\

$=lim_{\overline{k}(\theta,\phi)\rightarrow 0} sign(k_{1})(1+{k_{2}^{2}\over k_{1}^{2}}+{k_{3}^{2}\over k_{1}^{2}})lim_{\overline{k}(\theta,\phi)\rightarrow \overline{0},k_{1}\neq 0,k_{2}\neq 0,k_{3}\neq 0}k_{1}\mathcal{F}(f)(\overline{k})$\\

$=lim_{\overline{k}(\theta,\phi)\rightarrow 0} sign(k_{2})(1+{k_{1}^{2}\over k_{2}^{2}}+{k_{3}^{2}\over k_{2}^{2}})lim_{\overline{k}(\theta,\phi)\rightarrow \overline{0},k_{1}\neq 0,k_{2}\neq 0,k_{3}\neq 0}k_{2}\mathcal{F}(f)(\overline{k})$\\

$=lim_{\overline{k}(\theta,\phi)\rightarrow 0} sign(k_{3})(1+{k_{1}^{2}\over k_{3}^{2}}+{k_{2}^{2}\over k_{3}^{2}})lim_{\overline{k}(\theta,\phi)\rightarrow \overline{0},k_{1}\neq 0,k_{2}\neq 0,k_{3}\neq 0}k_{3}\mathcal{F}(f)(\overline{k})$\\

$=0$\\

as the cases $max(|k_{2}|,|k_{3}|)\leq |k_{1}|$, $max(|k_{1}|,|k_{3}|)\leq |k_{2}|$ and $max(|k_{1}|,|k_{2}|)\leq |k_{3}|$ are exhaustive.\\

Clearly, we can repeat the above arguments for ${\partial f\over \partial x_{i}}$, $1\leq i\leq 3$, and $f\in C^{42}(\mathcal{R}^{3})$, using the fact that ${\partial f\over \partial x_{i}}$ is of moderate decrease, in particularly of very moderate decrease, with the higher derivatives ${\partial^{l+m+n}{\partial f\over \partial x_{i}}\over \partial x^{l}y^{m}z^{n}}$ of moderate decrease $l+m+n+2$, in particularly of moderate decrease $l+m+n+1$.\\

For the next claim, we have, $\mathcal{F}(f)\in L^{1}(\mathcal{R}^{3})$, $(R)$, by Lemma \ref{integrablefirst}. A similar calculation shows that, as ${\partial f\over \partial x}$ is of moderate decrease $2$, that $f\in L^{{3\over 2}+\epsilon}(\mathcal{R}^{3})$, for $\epsilon>0$. Applying the Haussdorff-Young inequality, $\mathcal{F}({\partial f\over \partial x})\in L^{3-\delta}(\mathcal{R}^{3})$, for $\delta>0$. In particular, due to the decay again, $\mathcal{F}({\partial f\over \partial x})\in L^{2}(\mathcal{R}^{3})$. Locally, on $B(\overline{0},1)$, for $\delta>0$;\\

$\int_{B(\overline{0},1)}{1\over |\overline{k}|^{3-\delta}}d\overline{k}$\\

$=\int_{0\leq \theta\leq \pi,-\pi\leq \phi\leq \phi}\int_{0}^{1}{r^{2}\over r^{3-\delta}}dr d\theta d\phi$\\

$\leq 2\pi^{2}[r^{\delta}]^{1}_{0}$\\

$=2\pi^{2}<\infty$\\

so that ${1\over |\overline{k}|}\in L^{3-\delta}(B(\overline{0},1))$, in particularly ${1\over |\overline{k}|}\in L^{2}(B(\overline{0},1))$. As $\mathcal{F}({\partial f\over \partial x})\in L^{2}(B(\overline{0},1))$, by the Cauchy Schwarz inequality, we obtain that ${\mathcal{F}({\partial f\over \partial x})(\overline{k})\over |\overline{k}|}\in L^{1}(B(\overline{0},1))$, and by the decay, we have that ${\mathcal{F}({\partial f\over \partial x})(\overline{k})\over |\overline{k}|}\in L^{1}(\mathcal{R}^{3})$. Similar arguments show that ${\mathcal{F}({\partial f\over \partial x_{i}})(\overline{k})\over |\overline{k}|}\in L^{1}(\mathcal{R}^{3})$, for $1\leq i\leq 3$. We also have, using the fact that ${\partial^{2} f\over\partial x_{i}\partial x_{j}}$ is of moderate decrease and quasi split normal, $1\leq i\leq j\leq 3$, using the argument $(TT)$ twice, that for $k_{1}\neq 0$, $k_{2}\neq 0$, $k_{3}\neq 0$;\\

$\mathcal{F}({\partial^{2} f\over \partial x_{i}\partial x_{j}})=(ik_{i})(ik_{j})\mathcal{F}(f)(\overline{k})$\\

$=-k_{i}k_{j}\mathcal{F}(f)(\overline{k})$\\

so that;\\

${\mathcal{F}({\partial^{2} f\over \partial x_{i}\partial x_{j}})(\overline{k})\over |\overline{k}|^{2}}={-k_{i}k_{j}\over |\overline{k}|^{2}}\mathcal{F}(f)(\overline{k})$\\

with, for $k_{i}\neq 0$, $k_{j}\neq 0$;\\

$|{|-k_{i}k_{j}\over |\overline{k}|^{2}}|=|sign(k_{1})sign(k_{2}||{1\over (1+{k_{2}\over k_{1}}^{2}+{k_{3}\over k_{1}}^{2})^{1\over 2}}||{1\over (1+{k_{1}\over k_{2}}^{2}+{k_{3}\over k_{2}}^{2})^{1\over 2}}|\leq 1$\\

so that;\\

$|{\mathcal{F}({\partial^{2} f\over \partial x_{i}\partial x_{j}})(\overline{k})\over |\overline{k}|^{2}}|\leq |\mathcal{F}(f)(\overline{k})|$\\

and, by $(R)$, $\mathcal{F}(f)(\overline{k})\in L^{1}(\mathcal{R}^{3})$, so that ${\mathcal{F}({\partial^{2} f\over \partial x_{i}\partial x_{j}})(\overline{k})\over |\overline{k}|^{2}}\in L^{1}(\mathcal{R}^{3})$.\\

The last claim follows from the fact that, for $\overline{l}$, with $l_{1}\neq 0$, $l_{2}\neq 0$, $l_{3}\neq 0$, the translation $\mathcal{F}_{\overline{l}}({\partial f\over \partial x_{i}})(\overline{k})\in C^{1}(B(\overline{0},\epsilon'))$, for some $\epsilon'>0$. In particular, given $\epsilon>0$, there exists $\delta>0$, such that;\\

$max(|\int_{0}^{\delta}r\mathcal{F}_{\theta,\phi,\overline{l}}({\partial f\over \partial x_{j}})(r)dr|,|\int_{0}^{\delta}{d\over dr}(r\mathcal{F}_{\theta,\phi,\overline{l}}({\partial f\over \partial x_{j}})(r))dr|)<\epsilon$\\

uniformly in $\{\theta,\phi\}$.\\

\end{proof}

\end{document}